\documentclass[11pt]{amsart}
\usepackage[utf8]{inputenc}
\usepackage[titletoc]{appendix}

\usepackage[backend=biber, style=numeric, date=year, url=false, doi=false, isbn=false, eprint=true, firstinits=true, maxcitenames=5, maxbibnames=5]{biblatex}
\AtEveryBibitem{%
  \ifentrytype{online}{%
  }{%
    \clearfield{eprint}%
    \clearfield{urldate}%
  }%
}
\usepackage{blindtext}

\usepackage{scalerel,stackengine}
\stackMath
\newcommand\reallywidehat[1]{%
\savestack{\tmpbox}{\stretchto{%
  \scaleto{%
    \scalerel*[\widthof{\ensuremath{#1}}]{\kern.1pt\mathchar"0362\kern.1pt}%
    {\rule{0ex}{\textheight}}
  }{\textheight}%
}{2.4ex}}%
\stackon[-6.9pt]{#1}{\tmpbox}%
}
\usepackage{amssymb}
\usepackage{amsfonts}
\usepackage{amsthm}
\usepackage{tikz}
\usepackage{caption}
\usepackage{graphics}
\usepackage{xcolor}
\usepackage{shadethm}
\usepackage{subfigure}
\usepackage{epigraph}
\usepackage{placeins}
\usetikzlibrary{snakes}
\usepackage[new]{old-arrows}
\usepackage{pxfonts}

\usepackage{enumerate}
\usepackage{verbatim}
\usetikzlibrary{trees}
\usetikzlibrary{decorations.markings}
\usetikzlibrary{arrows}
\usetikzlibrary{cd}
\usepackage{mathrsfs}
\usepackage{xparse}
\usetikzlibrary{positioning,arrows,patterns}
\usetikzlibrary{decorations.markings}
\usetikzlibrary{calc}
\tikzset{
  photon/.style={decorate, decoration={snake}, draw=black},
  fermion/.style={draw=black, postaction={decorate},decoration={markings,mark=at position .55 with {\arrow{>}}}},
  fermion2/.style={dashed, dash phase=0.1pt, draw=black, postaction={decorate},decoration={markings,mark=at position .55 with {\arrow{>}}}},
  vertex/.style={draw,shape=circle,fill=black,minimum size=5pt,inner sep=0pt},
particle/.style={thick,draw=black},
particle2/.style={thick,draw=blue},
avector/.style={thick,draw=black, postaction={decorate},
    decoration={markings,mark=at position 1 with {\arrow[black]{triangle 45}}}},
gluon/.style={decorate, draw=black,
    decoration={coil,aspect=0}}
 }
\NewDocumentCommand\semiloop{O{black}mmmO{}O{above}}
{%
\draw[#1] let \p1 = ($(#3)-(#2)$) in (#3) arc (#4:({#4+180}):({0.5*veclen(\x1,\y1)})node[midway, #6] {#5};)
}

\usepackage{url}
\usepackage{hyperref}
\hypersetup{colorlinks=true, citecolor=purple, linktocpage, linkcolor=blue, urlcolor=cyan} 

\theoremstyle{plain}
\newtheorem{thm}{Theorem}[section]
\newtheorem{lem}[thm]{Lemma}
\newtheorem{prop}[thm]{Proposition}
\newtheorem{cor}[thm]{Corollary}

\theoremstyle{definition}
\newtheorem{defn}[thm]{Definition}

\newtheorem*{thm*}{Theorem}
\newtheorem*{lem*}{Lemma}
\newtheorem*{prop*}{Proposition}
\newtheorem*{cor*}{Corollary}
\newtheorem*{exe*}{Exercise}
\newtheorem*{defn*}{Definition}
\newtheorem{rem}[thm]{Remark}

\theoremstyle{remark}

\newcommand{\R}{\mathbb{R}}

\newcommand{\dd}{{\mathrm{d}}}

\newcommand{\id}{\mathrm{id}}

\DeclareMathOperator{\tr}{Tr}

\DeclareMathOperator{\Div}{\textnormal{div}}

\DeclareMathOperator{\gh}{gh}

\DeclareMathOperator{\ad}{ad}

\DeclareMathOperator{\Ad}{Ad}

\newcommand{\Der}{{\mathrm{Der}}}

\DeclareMathOperator{\End}{End}
\DeclareMathOperator{\Hom}{Hom}

\newcommand{\T}{\textsf{T}}

\newcommand{\de}{\partial}

\newcommand{\calA}{\mathcal{A}}
\newcommand{\calB}{\mathcal{B}}
\newcommand{\calH}{\mathcal{H}}
\newcommand{\calS}{\mathcal{S}}

\newcommand{\calO}{\mathcal{O}}
\newcommand{\calL}{\mathcal{L}}
\newcommand{\calM}{\mathcal{M}}

\newcommand{\calW}{\mathcal{W}}

\newcommand{\calN}{\mathcal{N}}

\newcommand{\calE}{\mathcal{E}}
\newcommand{\calP}{\mathcal{P}}
\newcommand{\calF}{\mathcal{F}}

\def\gpd{\,\lower1pt\hbox{$\longrightarrow$}\hskip-.24in\raise2pt
               \hbox{$\longrightarrow$}\,}

\let\Tilde=\widetilde

\let\Hat=\widehat

\DeclareMathOperator{\Map}{Map}

\newcommand{\I}{\mathrm{i}}

\newcommand{\ee}{\textnormal{e}}

\newcommand{\calV}{\mathcal{V}}

\newcommand{\Sym}{\textnormal{Sym}}

\makeatletter
\newcommand{\upint}{\DOTSI\upintop\ilimits@}
\newcommand{\upoint}{\DOTSI\upointop\ilimits@}
\makeatother

\usepackage{todonotes}

\makeatletter
\providecommand\@dotsep{5}
\renewcommand{\listoftodos}[1][\@todonotes@todolistname]{%
  \@starttoc{tdo}{#1}}
\makeatother

\makeatletter

\pgfdeclareshape{genus}{
 \anchor{center}{\pgfpointorigin}
\backgroundpath{
    \begingroup
    \tikz@mode
    \iftikz@mode@fill

         \pgfpathmoveto{\pgfqpoint{-0.78cm}{-.17cm}}
     \pgfpathcurveto %
            {\pgfpoint{-0.35cm}{-.44cm}}
        {\pgfpoint{0.35cm}{-.44cm}}
        {\pgfpoint{.78cm}{-0.17cm}} 
     \pgfpathmoveto{\pgfqpoint{-0.78cm}{-0.17cm}}
     \pgfpathcurveto %
            {\pgfpoint{-0.25cm}{.25cm}}
        {\pgfpoint{.25cm}{.25cm}}
        {\pgfpoint{0.78cm}{-0.17cm}}
        \pgfusepath{fill}
        \fi

        \iftikz@mode@draw
     \pgfpathmoveto{\pgfqpoint{-1cm}{0cm}}
     \pgfpathcurveto %
            {\pgfpoint{-0.5cm}{-.5cm}}
        {\pgfpoint{0.5cm}{-.5cm}}
        {\pgfpoint{1cm}{0cm}}

     \pgfpathmoveto{\pgfqpoint{-0.75cm}{-0.15cm}}
     \pgfpathcurveto %
            {\pgfpoint{-0.25cm}{.25cm}}
        {\pgfpoint{.25cm}{.25cm}}
        {\pgfpoint{0.75cm}{-0.15cm}}
              \pgfusepath{stroke}
        \fi
        \endgroup
    }
    }

    \makeatother

\title[Formal Global AKSZ Gauge Observables and Wilson Surfaces]{Formal Global AKSZ Gauge Observables and Generalized Wilson Surfaces}
\author[N. Moshayedi]{Nima Moshayedi}
\address{Institut f\"ur Mathematik\\ Universit\"at Z\"urich\\ 
Winterthurerstrasse 190
CH-8057 Z\"urich}
\email[N.~Moshayedi]{nima.moshayedi@math.uzh.ch}
\thanks{This research was supported by the NCCR SwissMAP, funded by the Swiss National Science
Foundation, and by the SNF grant No. 200020\_192080.}

\begin{document}

\maketitle

\begin{abstract}
    We consider a construction of observables by using methods of supersymmetric field theories.
    In particular, we give an extension of AKSZ-type observables constructed in \cite{Mn3} using the Batalin--Vilkovisky structure of AKSZ theories to a formal global version with methods of formal geometry. We will consider the case where the AKSZ theory is ``split'' which will give an explicit construction for formal vector fields on base and fiber within the formal global action.
    Moreover, we consider the example of formal global generalized Wilson surface observables whose expectation values are invariants of higher-dimensional knots by using $BF$ field theory. These constructions give rise to interesting global gauge conditions such as the differential Quantum Master Equation and further extensions.
\end{abstract}

\tableofcontents

\section{Introduction}
Observables play a fundamental role in theoretical and mathematical physics. They are used in several constructions, e.g. deformation quantization and factorization algebras. 
In \cite{Mn3}, a method for constructing observables in the setting of AKSZ theories was introduced, where several examples, including Wilson loop type observables for different theories, have been addressed.  

These constructions were given using the approach of supersymmetric field theory and methods of functional integrals. In particular, the focus lies within a special formalism dealing with gauge theories which is called the Batalin--Vilkovisky (BV) formalism. This formalism was developed by Batalin and Vilkovisky in a series of papers \cite{BV1,BV2,BV3} during the 1970's and 1980's in order to deal with the functional integral quantization approach where the Lagrangian is invariant under certain symmetries and the integral is ill-defined. They have shown (later also formulated in a more mathematical language by Schwarz) that these issues can be resolved by replacing the ill-defined integral by a well-defined (after some regularization is also introduced) one without changing the final value. The mathematical structures of this powerful formalism have been studied since then by many different people. 

AKSZ theories \cite{AKSZ} (named after Alexandrov, Kontsevich, Schwarz and Zaboronsky) are a particular type of field theories where the space of fields is given by a mapping space between manifolds. It can be shown that these theories, regarded in a special setting, will give rise to field theories as formulated in the BV setting. Many interesting theories are in fact of AKSZ-type, e.g. Chern--Simons theory \cite{Witten1989,AS,AS2,Cattaneo2008,CMW}, the Poisson sigma model \cite{I,SS1,CF4}, Rozansky--Witten theory \cite{RozanskyWitten1997}, the Courant sigma model \cite{Roytenberg2005}, $BF$ theory \cite{Mn,CMR1}, Witten's $A$- and $B$-twisted sigma models \cite{Witten1988a} and $2D$ Yang--Mills theory \cite{IM}.

The globalization idea originates from a field theoretic approach to globalization of Kontsevich's star product \cite{K} in deformation quantization. The associated field theory is given by the Poisson sigma model. 
The Poisson sigma model is a 2-dimensonal bosonic string theory with target a Poisson manifold which was first considered by Ikeda \cite{I} and Schaller--Strobl \cite{SS1} by the attempt of studying 2D gravity theories and combine them to a common form with Yang--Mills theories.
Using the Poisson sigma model on the disk, Cattaneo and Felder have proven that Kontsevich's star product is exactly given by the perturbative expansion of its functional integral quantization \cite{CF1}. Regarding the fact that the Poisson sigma model is a gauge theory, it is interesting to note that it is a fundamental non-trivial theory where the BRST gauge formalism \cite{BRS1,BRS2,BRS3,Tyutin1976} does not work if the Poisson structure is not linear. In fact, to treat the Poisson sigma model and its quantization, one has to use the BV formalism.
However, the field theoretic construction of Kontsevich's star product was only considered locally since Kontsevich's formula was only given for the local picture on the upper half-plane. Later on, using techniques of formal geometry, developed by e.g. Gelfand--Fuks \cite{GelfandFuks1969,GelfandFuks1970}, Gelfand--Kazhdan \cite{GK} or Bott \cite{B}, it was possible to construct a globalization, similar to the approach of Fedosov for symplectic manifolds which only covers the case of constant (symplectic) Poisson structures \cite{Fedosov1994}. 

In \cite{CF3,BCM} this approach was first extended to the field theoretic BV construction of the Poisson sigma model for closed source manifolds. 
In recent work \cite{CMW3} this construction was extended to the case of source manifolds with boundary. There one has to extend the BV formalism to the \emph{BV-BFV formalism} which couples the boundary BFV theory to the bulk BV theory such that everything is consistent in the cohomological formalism. Here BFV stands for Batalin--Fradkin--Vilkovisky which formulated a Hamiltonian version of the BV construction in \cite{BF1,FV1}. The bulk-boundary coupling (the BV-BFV formalism) was first introduced classically in \cite{CMR1,CMR3} and extended to the quantum version in \cite{CMR2}.
The globalization construction for the Poisson sigma model on manifolds with boundary was more generally extended in \cite{CMW4} to a special class of AKSZ theories which are called ``split'' where the case of the Poisson sigma model is an example.

The aim of this paper is to extend the constructions of \cite{Mn3} to a formal global construction. In fact we will construct formal global observables by using the notion of a Hamiltonian $Q$-bundle \cite{KotovStrobl2015} together with notions of formal geometry, and we will study the formal global extension of Wilson loop type observables for the Poisson sigma model. 

Additionally, we discuss the formal global extension of Wilson surface observables which have been studied in \cite{CattRoss2005} by using the AKSZ formulation of $BF$ theories. We will show that these constructions lead to interesting gauge conditions such as the \emph{differential Quantum Master Equation} (and further extensions).

These constructions are expected to extend to manifolds with boundary by using the BV-BFV formalism as the globalization constructions have been studied for nonlinear split AKSZ theories on manifolds with boundary \cite{CMW4}.  
\\

{\textbf{Acknowledgements:}} 
I would like to thank Alberto Cattaneo for useful comments and remarks on a first draft of these notes.

\section{The Batalin--Vilkovisky (BV) formalism}
In this section we will recall some aspects of the Batalin--Vilkovisky formalism as in \cite{CMR2,Mnev2019}. An introductory reference for learning about the formalism is \cite{CattMosh1}, which also covers the most important concepts of supergeometry and the case of manifolds with boundary (BV-BFV). 

\subsection{Classical BV picture} Let us start with the classical setting of the BV formalism.

\begin{defn}[BV manifold]
\label{defn:BV_manifold}
A \emph{BV manifold} is a triple $$(\calF,\calS,\omega)$$ such that $\calF$ is a $\mathbb{Z}$-graded supermanifold\footnote{Typically, this is an infinite-dimensional manifold. However, there are certain cases where this is a finite-dimensional manifold, e.g. if we consider the moduli of flat connections on a compact, oriented 2-manifold with holonomies on the boundary according to Atiyah and Bott \cite{AtiyahBott1983} which is of importance regarding $BF$ theory.}, $\calS$ is an even function on $\calF$ of degree 0 and $\omega$ is an odd symplectic form on $\calF$ of degree $-1$. Moreover, we want that $\calS$ satisfies the Classical Master Equation (CME) 
\begin{equation}
\boxed{
    \{\calS,\calS\}_\omega=0,}
\end{equation}
where $\{\enspace,\enspace\}_{\omega}$ denotes the odd Poisson bracket induced by the odd symplectic form $\omega$. This odd Poisson bracket is also called \emph{BV bracket} and, according to Batalin and Vilkovisky, is often denoted by round brackets $(\enspace,\enspace)$. We will call $\calF$ the \emph{BV space of fields}\footnote{Usually, the BV space of fields is given by the $(-1)$-shifted cotangent bundle of the BRST space of fields, i.e. $\calF_{\mathrm{BV}}=T^*[-1]\calF_\mathrm{BRST}$}, $\calS$ the \emph{BV action} (sometimes also called the \emph{master action}) and $\omega$ the \emph{BV symplectic form}.
\end{defn}
\begin{rem}
In physics, the $\mathbb{Z}$-grading is called the \emph{ghost number}. We will denote the ghost number by $\mathrm{gh}$ and the form degree by $\deg$.
\end{rem}
\begin{rem}
The data of a BV manifold induces a \emph{symplectic cohomological} vector field $Q$ of degree $+1$ which is given by the Hamiltonian vector field of $\calS$, i.e. 
\begin{equation}
\label{eq:CME}
    \iota_Q\omega=\delta\calS,
\end{equation}
wher $\delta$ denotes the de Rham differential on $\calF$. The cohomological property means that $[Q,Q]=0$ and the symplectic property means $L_Q\omega=0$, where $L$ denotes the Lie derivative. Moreover, note that by definition $$Q=\{\calS,\enspace\}_\omega.$$
\end{rem}
\begin{defn}[Exact BV manifold]
\label{defn:exact_BV_manifold}
A BV manifold is called \emph{exact} if $\omega=\delta\alpha$ for some primitive 1-form $\alpha$.
\end{defn}
In what will follow, we will mostly consider exact BV manifolds. According to the use of sigma models we want to consider space-time manifolds as the source manifolds for our theory. Moreover, in this paper we will restrict ourself to topological theories. 
\begin{defn}[BV theory]
\label{defn:BV_theory}
A \emph{BV theory} is an assignment of a manifold $\Sigma$ to a BV manifold 
\begin{equation}
    \Sigma\mapsto(\calF_\Sigma,\calS_\Sigma,\omega_\Sigma,Q_\Sigma).
\end{equation}
\end{defn}

\subsection{Quantum BV picture}
We continue with the quantum setting of the BV formalism.

\begin{defn}[Quantum BV manifold]
\label{defn:quantum_BV_manifold}
A \emph{quantum BV manifold} is a quadruple $(\calF,\omega,\mu,\calS)$ such that $\calF$ is a $\mathbb{Z}$-graded supermanifold, $\omega$ a symplectic form on $\calF$ of degree $-1$, $\mu$ a volume element\footnote{We want the space of fields $\calF$ to be endowed with a natural measure.} of $\calF$ which is compatible with $\omega$ in the sense that the associated BV Laplacian 
\begin{equation}
    \Delta\colon f\mapsto \frac{1}{2}\Div_{\mu}\{f,\enspace\}_\omega
\end{equation}
satisfies 
\begin{equation}
    \Delta^2=0,
\end{equation}
and $\calS$ is a degree 0 function on $\calF$ such that it satisfies the QME \eqref{eq:QME}.
\end{defn}

\begin{rem}
The BV Laplacian satisfies a generalized BV Leibniz rule. For two functions $f,g$ on $\calF$ we have
\[
\Delta(fg)=\Delta(f)g\pm f\Delta(g)\pm\{f,g\}_\omega.
\]
see also \cite{Khudaverdian2004,Severa2006} for a mathematical exposure to the origin of the BV Laplacian.  
\end{rem}

Moreover, define $\delta_{\mathrm{BV}}$ to be the degree $+1$ operator given by 
\begin{equation}
    \delta_{\mathrm{BV}}:=Q-\I\hbar\Delta
\end{equation}
which satisfies 
\begin{equation}
    \delta_{\mathrm{BV}}^2=0.
\end{equation}

The following theorem is one of the main statements in the formalism developed by Batalin and Vilkovisky. In its present form it was stated by Schwarz on general manifolds \cite{S}.
\begin{thm}[Batalin--Vilkovisky]
\label{thm:BV}
    For any half-density $f$ on $\calF$ we have: 
    \begin{enumerate}
        \item If $f=\Delta g$, then $$\int_{\calL}f=0,$$ for a Lagrangian submanifold $\calL\subset\calF$,
        \item If $\Delta f=0$, then $$\frac{\dd}{\dd t}\int_{\calL_t}f=0,$$ for any continuous family $(\calL_t)$ of Lagrangian submanifolds of $\calF$.
    \end{enumerate}
\end{thm}

\begin{rem}
The choice of Lagrangian submanifold is in fact equivalent to fixing a gauge. 
The second part of Theorem \ref{thm:BV} tells us that if we have an integral over a Lagrangian submanifold which is ill-defined, but on the other hand $\Delta f=0$, then we can deform the Lagrangian submanifold $\calL$ continuously to a Lagrangian submanifold $\calL'$ (choosing a different gauge) where the integral is well-defined. In application to quantum field theory, we have $f=\ee^{\frac{\I}{\hbar}\calS}$. Hence, for gauge-independence, we need to impose 
\begin{equation}
\label{eq:QME}
\boxed{
    \Delta\ee^{\frac{\I}{\hbar}\calS}=0\Longleftrightarrow \{\calS,\calS\}_\omega-2\I\hbar\Delta\calS=0.}
\end{equation}
The condition \eqref{eq:QME} is called the \emph{Quantum Master Equation (QME)}. If we let $\calS$ depend on $\hbar$, we can see that in order zero we get the CME $\{\calS,\calS\}_\omega=0$. One can then solve \eqref{eq:QME} order by order.
\end{rem}

\subsection{$L_\infty$-structure}
Recall that a $Q$-manifold with trivial body induces an $L_\infty$-algebra structure (see e.g. \cite{MehtaZambon2012}). More
generally, a $Q$-manifold with non-trivial body induces an $L_\infty$-algebroid structure. Similarly, a BV manifold endowed with its $Q$-structure induces an $L_\infty$-algebra structure on $\calF$ \cite{Stasheff1997_2}. This $L_\infty$-algebra encodes all the relevant classical information of the field theory. Hence, at the classical level, Lagrangian field theories can be equivalently described in terms of the underlying (cyclic\footnote{A \emph{cyclic} $L_\infty$-algebra \cite{Kontsevich1994} is an $L_\infty$-algebra $\mathfrak{g}$ endowed with a non-degenerate, symmetric, bilinear pairing $\langle\enspace,\enspace\rangle_\mathfrak{g}\colon \mathfrak{g}\oplus \mathfrak{g}\to \R$ such that 
\[
\langle X_1,\ell_{n+1}(X_2,\ldots,X_{n+1})\rangle_\mathfrak{g}=(-1)^{n+n(\deg(X_1)+\deg(X_{n+1}))+\deg(X_{n+1})\sum_{j=1}^n\deg(X_j)}\langle X_{n+1},\ell_{n}(X_1,\ldots,X_n)\rangle_\mathfrak{g},
\]
for $X_1,\ldots,X_{n+1}\in\mathfrak{g}$ and where $(\ell_n)$ denote the $n$-ary brackets on $\mathfrak{g}$. 
In the case of a $Q$-manifold the cyclic inner product corresponds to a symplectic structure.}) $L_\infty$-algebra structure. Moreover, equivalent theories induce quasi-isomorphic $L_\infty$-algebras. The unary operation $\ell_1$ is in fact encoded in the linear part of the action $Q=\{\calS,\enspace\}_\omega$ on the field corresponding to the image of $\ell_1$. The higher brackets then make the linearized expressions covariant and to allow for higher interaction terms. The operator $\delta_\mathrm{BV}$ in fact induces a quantum $L_\infty$-algebra (or loop homotopy algebra) on the same graded space. In particular, by a direct application of the homological perturbation lemma, one can prove a similar decomposition theorem and compute its \emph{minimal model} as for the classical case, which leads directly to a homotopy between a quantum $L_\infty$-algebra and its minimal model in which the non-triviality of the action is fully absorbed in the higher brackets. Moreover, the homotopy Maurer--Cartan theory\footnote{This is the theory induced by the action term $\calS_\mathrm{MC}(\Psi)=\sum_{j\geq1}\frac{1}{(j+1)!}\langle \Psi,\ell_j(\Psi,\ldots,\Psi)\rangle_\mathfrak{g}$ for a cyclic $L_\infty$-algebra $\mathfrak{g}$ endowed with an inner product $\langle\enspace,\enspace\rangle_\mathfrak{g}$.
Here $(\ell_j)$ denotes the family of $j$-ary brackets on $\mathfrak{g}$.
The stationary locus of this action is given by solutions of the homotopy Maurer--Cartan equation $\sum_{j\geq 1}\frac{1}{j!}\ell_j(\Psi,\ldots,\Psi)=0$. In fact, the deformed Lagrangian (still classical) 
\[
\calS(\Psi)=\frac{1}{2}\langle \Psi,Q(\Psi)\rangle_\mathfrak{g}+\sum_{j\geq1}\frac{1}{(j+1)!}\langle \Psi,\ell_j(\Psi,\ldots,\Psi)\rangle_\mathfrak{g}
\] satisfies the CME.} implies that for an arbitrary $L_\infty$-algebra the BV complex of fields, ghosts and anti fields is just the $L_\infty$-algebra itself. See e.g. \cite{Stasheff1997,Stasheff1997_2,JRSW2018} for a more detailed discussion of $L_\infty$-structures for BV field theories. 


\section{AKSZ theories}
\label{sec:AKSZ_Theories}

\subsection{Preliminaries}
In \cite{AKSZ}, Alexandrov, Kontsevich, Schwarz and Zaboronsky have proposed a class of local field theories which are compatibel with the Batalin--Vilkovisky gauge formalism construction, in the sense that the constructed local actions are solutions to the Classical Master Equation. Hence, these theories give a subclass of BV theories. In this section we want to recall the most important notions of AKSZ sigma models. We start with defining the ingredients.

\begin{defn}[Differential graded symplectic manifold]
\label{defn:dg_symplectic_manifold}
A \emph{differential graded symplectic manifold} of degree $k$ is a triple $$(\calM,\Theta_\calM,\omega_\calM=\dd_{\calM}\alpha_\calM)$$ such that $\calM$ is a $\mathbb{Z}$-graded manifold, $\Theta_\calM\in C^\infty(\calM)$ is a function on $\calM$ of degree $k+1$, and $\omega_\calM\in \Omega^2(\calM)$ is an exact symplectic form of degree $k$ with primitive 1-form $\alpha_\calM\in \Omega^1(\calM)$, such that 
\begin{equation}
    \left\{\Theta_\calM,\Theta_\calM\right\}_{\omega_\calM}=0,
\end{equation}
where $\{\enspace,\enspace\}_{\omega_\calM}$ is the odd Poisson bracket induced by $\omega_\calM$. We have denoted by $\dd_{\calM}$ the de Rham differential on $\calM$.
\end{defn}

\begin{rem}
We denote by $Q_\calM\in\mathfrak{X}(\calM)$ the Hamiltonian vector field of $\Theta_\calM$, defined by the equation $$\iota_{Q_\calM}\omega_\calM=\dd_{\calM}\Theta_\calM$$ with the properties $[Q_\calM,Q_\calM]=0$ (cohomological) and $L_{Q_\calM}\omega_\calM=0$ (symplectic). Note that $Q_\calM$ is of degree $+1$.
A quadruple $(\calM,Q_\calM,\Theta_\calM,\omega_\calM=\dd_{\calM}\alpha_\calM)$ as in Definition \ref{defn:dg_symplectic_manifold} is also called a \emph{Hamiltonian $Q$-manifold}.
\end{rem}

\subsection{AKSZ sigma models}
\label{sec:AKSZ_sigma_models}
Let $\Sigma_d$ be a $d$-dimensional compact, oriented manifold (possibly with boundary) and consider its shifted tangent bundle $T[1]\Sigma_d$. Moreover, fix a Hamiltonian $Q$-manifold $$(\calM,Q_{\calM},\Theta_{\calM},\omega_{\calM}=\dd_{\calM}\alpha_{\calM})$$ of degree $d-1$ for $d\geq 0$. We can consider the mapping space of graded manifolds from $T[1]\Sigma_d$ to $\calM$ to be our space of fields:
\begin{equation}
    \calF^\calM_{\Sigma_d}:=\Map_{\mathrm{GrMnf}}(T[1]\Sigma_d,\calM),
\end{equation}
where $\Map_{\mathrm{GrMnf}}$ denotes the mapping space between graded manifolds\footnote{More precisely, $\Map_{\mathrm{GrMnf}}$ denotes the right adjoint functor to the Cartesian product in the category of graded manifolds with a fixed factor. On objects $X,Y,Z$ we have $\Hom(X,\Map_{\mathrm{GrMnf}}(Y,Z))=\Hom(X\times Y,Z)$, where $\Hom$ denotes the set of graded manifold morphisms.}.
We would like to endow $\calF^\calM_{\Sigma_d}$ with a $Q$-manifold structure. This can be done by considering the lifts of the de Rham differential $\dd_{\Sigma_d}$ on $\Sigma_d$ and the cohomological vector field $Q_\calM$ on the target $\calM$ to the mapping space. Hence, we get a cohomological vector field 
\begin{equation}
    Q_{\Sigma_d}:=\Hat{\dd}_{\Sigma_d}+\Hat{Q}_{\calM}\in \mathfrak{X}\left(\calF^\calM_{\Sigma_d}\right),
\end{equation}
where $\Hat{\dd}_{\Sigma_d}$ and $\Hat{Q}_{\calM}$ denote the corresponding lifts to the mapping space. Note that we can regard $\dd_{\Sigma_d}$ as a cohomological vector field on $T[1]\Sigma_d$.
Consider the following push-pull diagram 
\begin{equation}
    \calF^\calM_{\Sigma_d}\xleftarrow{\mathrm{p}} \calF^\calM_{\Sigma_d}\times T[1]\Sigma_d\xrightarrow{\mathrm{ev}} \calM,
\end{equation}
where $\mathrm{p}$ denotes the projection onto $\calF^\calM_{\Sigma_d}$ and $\mathrm{ev}$ is the evaluation map. We can construct a \emph{transgression} map 
\begin{equation}
    \mathscr{T}_{\Sigma_d}:=\mathrm{p}_*\mathrm{ev}^*\colon \Omega^\bullet(\calM)\to \Omega^\bullet\left(\calF^\calM_{\Sigma_d}\right).
\end{equation}
Note that the map $\mathrm{p}_*$ is given by fiber integration on $T[1]\Sigma_d$. Now we can endow the space of fields $\calF^\calM_{\Sigma_d}$ with a symplectic structure $\omega_{\Sigma_d}$ by setting
\begin{equation}
    \omega_{\Sigma_d}:=(-1)^{d}\mathscr{T}_{\Sigma_d}(\omega_{\calM})\in \Omega^2\left(\calF^\calM_{\Sigma_d}\right).
\end{equation}
Moreover, we will get a solution $\calS_{\Sigma_d}$ to the CME, the BV action functional, by 
\begin{equation}
    \calS_{\Sigma_d}:=\underbrace{\iota_{\Hat{\dd}_{\Sigma_d}}\mathscr{T}_{\Sigma_d}(\alpha_{\calM})}_{=:\calS^\mathrm{kin}_{\Sigma_d}}+\underbrace{\mathscr{T}_{\Sigma_d}(\Theta_{\calM})}_{=:\calS^\mathrm{target}_{\Sigma_d}}\in C^\infty\left(\calF^\calM_{\Sigma_d}\right).
\end{equation}
Indeed, one can check that 
\begin{equation}
    \left\{\calS_{\Sigma_d},\calS_{\Sigma_d}\right\}_{\omega_{\Sigma_d}}=0.
\end{equation}
Note that the symplectic form $\omega_{\Sigma_d}$ is of degree $(d-1)-d=-1$ as expected. Moreover, the action $\calS_{\Sigma_d}$ is of degree $0$. Thus this setting does indeed induce a BV manifold $\left(\calF^\calM_{\Sigma_d},\calS_{\Sigma_d},\omega_{\Sigma_d}\right)$. Consider local coordinates $(x^\mu)$ on $\calM$ and let $(u^i)$ be local coordinates on $\Sigma_d$ for $1\leq i\leq d$. Denote the odd fiber coordinates of degree $+1$ on $T[1]\Sigma_d$ by $\theta^i=\dd_{\Sigma_d} u^i$. Then, for a field $\calA\in \calF^\calM_{\Sigma_d}$, we have the local expression
\begin{equation}
    \calA^\mu(u,\theta)=\sum_{\ell=0}^d\,\,\underbrace{\sum_{1\leq i_1<\dotsm <i_\ell\leq d}\calA^\mu_{i_1\ldots i_\ell}(u)\theta^{i_1}\land\dotsm \land \theta^{i_\ell}}_{\calA^\mu_{(\ell)}(u,\theta)}\in \bigoplus_{\ell=0}^d C^\infty(\Sigma_d)\otimes \bigwedge^\ell T^*\Sigma_d.
\end{equation}
The functions $\calA^\mu_{i_1\ldots i_\ell}\in C^\infty(\Sigma_d)$ are of degree $\deg(x^\mu)-\ell$ on $\calF^\calM_{\Sigma_d}$. The local expression of the symplectic form $\omega_{\calM}$ and its primitive 1-form $\alpha_{\calM}$ on $\calM$ are given by 
\begin{align}
    \alpha_{\calM}&=\alpha_{\mu}(x)\dd_{\calM}x^\mu\in \Omega^1(\calM),\\
    \omega_{\calM}&=\frac{1}{2}\omega_{\mu_1\mu_2}(x)\dd_{\calM}x^{\mu_1}\land \dd_{\calM}x^{\mu_2}\in \Omega^2(\calM).
\end{align}
Locally, using the expressions above, we get the following expression for the BV symplectic form, its primitive 1-form and the BV action functional:
\begin{align}
    \alpha_{\Sigma_d}&=\int_{\Sigma_d}\alpha_\mu(\calA)\delta\calA^\mu\in \Omega^1\left(\calF^\calM_{\Sigma_d}\right),\\
    \omega_{\Sigma_d}&=(-1)^d\frac{1}{2}\int_{\Sigma_d}\omega_{\mu_1\mu_2}(\calA)\delta\calA^{\mu_1}\land \delta\calA^{\mu_2}\in \Omega^2\left(\calF^\calM_{\Sigma_d}\right),\\
    \calS_{\Sigma_d}&=\int_{\Sigma_d}\alpha_\mu(\calA)\dd_{\Sigma_d}\calA^\mu+\int_{\Sigma_d}\Theta_{\calM}(\calA)\in C^\infty\left(\calF^\calM_{\Sigma_d}\right).
\end{align}
Note that we have denoted by $\delta$ the de Rham differential on $\calF^\calM_{\Sigma_d}$. If we consider Darboux coordinates on $\calM$, we get that \[\omega_\calM=\frac{1}{2}\omega_{\mu_1\mu_2}\dd_\calM x^{\mu_1}\land \dd_{\calM} x^{\mu_2},\] where the $\omega_{\mu_1\mu_2}$ are constant implying that $\alpha_\calM=\frac{1}{2}x^{\mu_1}\omega_{\mu_1\mu_2}\dd_\calM x^{\mu_2}$. Hence we get the BV symplectic form
\begin{equation}
    \omega_{\Sigma_d}=\frac{1}{2}\int_{T[1]\Sigma_d}\mu_{\Sigma_d}\left(\omega_{\mu_1\mu_2}\delta\calA^{\mu_1}\land\delta\calA^{\mu_2}\right)=\frac{1}{2}\int_{\Sigma_d}\left(\omega_{\mu_1\mu_2}\delta\calA^{\mu_1}\land\delta\calA^{\mu_2}\right)^\mathrm{top}
\end{equation}
and the master action 
\begin{equation}
    \calS_{\Sigma_d}=\int_{T[1]\Sigma_d}\mu_{\Sigma_d}\left(\frac{1}{2}\calA^\mu\omega_{\mu_1\mu_2}\boldsymbol{D}_{\Sigma_d}\calA^{\mu_2}\right)+(-1)^d\int_{T[1]\Sigma_d}\mu_{\Sigma_d}\calA^*\Theta_\calM,
\end{equation}
where $\mu_{\Sigma_d}$ is a canonical measure on $T[1]\Sigma_d$ and $\boldsymbol{D}_{\Sigma_d}=\theta^j\frac{\de}{\de u_j}$ the superdifferential on $T[1]\Sigma_d$.

\section{Hamiltonian $Q$-bundles}
We want to construct a combination of the notion of $Q$-manifolds and the concept of Hamiltonian vector fields together with the notion of vector bundles, where we want to extend most of our constructions on the fiber (see also \cite{KotovStrobl2015}). We will see that the fiber will represent the target of an AKSZ theory for an embedded source manifold when lifted to an AKSZ-BV theory. We will call the fiber theory \emph{auxiliary}. In this section we will give the main definitions as in \cite{Mn3}. Let us start with the definition of the trivial case.  

\begin{defn}[Trivial $Q$-bundle]
\label{defn:trivial_Q_bundle}
Let $\calN$ be a graded manifold and $(\calM,Q_\calM)$ a graded $Q$-manifold. A \emph{trivial $Q$-bundle} is a trivial bundle 
\begin{equation}
    \pi\colon \calE:=\calM\times \calN\to \calM
\end{equation}
such that $\dd\pi(Q_\calE)=Q_\calM$, where $Q_\calE$ denotes the $Q$-structure on the total space $\calE$.
\end{defn}

\begin{rem}
Note that this implies that 
\[Q_\calE=Q_\calM+\calV,\]
where $\calV\in \ker \dd\pi\cong C^\infty(\calM)\Hat{\otimes}\mathfrak{X}(\calN)$ denotes the vertical part of $Q_\calE$. The fact that $[Q_\calE,Q_\calE]=0$ can be translated to 
\begin{equation}
    \underbrace{[Q_\calM,Q_\calM]}_{=0}+[Q_\calM,\calV]+\frac{1}{2}[\calV,\calV]=0.
\end{equation}
\end{rem}

\begin{defn}[Trivial Hamiltonian $Q$-bundle]
\label{defn:trivial_Hamiltonian_Q-bundle}
A \emph{trivial Hamiltonian $Q$-bundle} of degree $n\in\mathbb{Z}$ is a trivial $Q$-bundle 
\[\pi\colon \calE:=\calM\times \calN\to \calM\] 
as in Definition \ref{defn:trivial_Q_bundle} with $Q_\calE=Q_\calM+\calV$ such that the fiber $\calN$ is endowed with an exact symplectic structure $\omega_\calN=\dd_\calN\alpha_\calN\in \Omega^2(\calN)$ of degree $n$ with $\alpha_\calN\in\Omega^1(\calN)$ and a Hamiltonian function $\Theta_\calE\in C^\infty(\calE)$ of degree $n+1$ satisfying
\begin{align}
    &\calV=\{\Theta_\calE,\enspace\}_{\omega_\calN}\\
    \label{eq:classical_1}
    &Q_\calM(\Theta_\calE)+\frac{1}{2}\{\Theta_\calE,\Theta_\calE\}_{\omega_\calN}=0.
\end{align}
\end{defn}

We can now give the definition of a general Hamiltonian $Q$-bundle.
\begin{defn}[Hamiltonian $Q$-bundle]
\label{defn:Hamiltonian_Q-bundle}
A \emph{Hamiltonian $Q$-bundle} is a $Q$-bundle $\pi\colon\calE\to \calM$ where the total space $\calE$ is endowed with a degree $n$ exact pre-symplectic form $\omega_\calE=\dd_\calE\alpha_\calE$ such that $\ker\omega_\calE\subset T\calE$ is transversal to the vertical distribution $T^\mathrm{vert}\calE$ and hence $\ker\omega_\calE$ defines a flat Ehresmann connection $\nabla_{\omega_\calE}$. Moreover, there is a Hamiltonian function $\Theta_\calE\in C^\infty(\calE)$ with $$\iota_{Q_\calE}\omega_\calE=\dd^\mathrm{vert}_\calE\Theta_\calE,$$ where $\dd_{\calE}^\mathrm{vert}$ denotes the vertical part of the de Rham differential on $\calE$ as a pullback by the natural inclusion $T^\mathrm{vert}\calE\hookrightarrow T\calE$. Finally, we also want that 
\begin{equation}
    \left(Q^\mathrm{hor}_\calE+\frac{1}{2}Q^\mathrm{vert}_\calE\right)(\Theta_\calE)=0,
\end{equation}
where we split $Q_\calE=Q^\mathrm{hor}_\calE+Q^\mathrm{vert}_\calE$ into its horizontal and vertical parts by using the Ehresmann connection $\nabla_{\omega_\calE}$ defined by $\omega_\calE$.
\end{defn}

\section{Observables in the BV formalism}
We want to define certain classes of observables arising within the BV construction which are compatible with the structure of an underlying $Q$-bundle. We will start with the classical setting.

\subsection{Observables for classical BV manifolds}

\begin{defn}[BV classical observable]
A \emph{classical observable} for a BV manifold $(\calF,\calS,Q,\omega)$ is defined as a function $\calO\in C^\infty(\calF)$ of degree 0 such that 
\begin{equation}
    Q(\calO)=0.
\end{equation}
\end{defn}

\begin{defn}[Equivalence of BV classical observables]
Two BV classical observables $\calO$ and $\Tilde{\calO}$ are said to be equivalent if 
\begin{equation}
    \Tilde{\calO}-\calO=Q(\Psi),\quad \Psi\in C^\infty(\calF),
\end{equation}
or equivalently, $\calO$ and $\Tilde{\calO}$ have the same $Q$-cohomology class.
\end{defn}

\begin{defn}[BV classical pre-observable]
For a classical BV theory $$(\calF,\calS,Q,\omega)$$ we define a \emph{pre-observable} to be a Hamiltonian $Q$-bundle over $\calF$ of degree $-1$. We denote the fiber by $\calF^{\mathrm{aux}}$ and call them the \emph{space of auxiliary fields}, which itself is endowed with a symplectic structure $\omega^{\mathrm{aux}}$ of degree $-1$ and an action functional $\calS^\mathrm{aux}\in C^\infty(\calF\times \calF^\mathrm{aux})$ of degree 0 such that 
\begin{equation}
    Q(\calS^\mathrm{aux})+\frac{1}{2}\{\calS^\mathrm{aux},\calS^\mathrm{aux}\}_{\omega^\mathrm{aux}}=0.
\end{equation}
\end{defn}

Using the notions of quantum BV manifolds as in Definition \ref{defn:quantum_BV_manifold}, we can define a fiber auxiliary version which is compatible with the Hamiltonian $Q$-bundle construction as in Definition \ref{defn:Hamiltonian_Q-bundle}.

\begin{defn}[BV semi-quantum pre-observable]
For a classical BV theory $(\calF,\calS,Q,\omega)$ we define a \emph{BV semi-quantum pre-observable} to be a quadruple $$(\calF^\mathrm{aux},\calS^\mathrm{aux},\omega^\mathrm{aux},\mu^\mathrm{aux})$$ such that $\mu^\mathrm{aux}$ is a volume form on $\calF^\mathrm{aux}$ compatible with $\omega^\mathrm{aux}$, i.e. the associated BV Laplacian on $C^\infty(\calF^\mathrm{aux})$ given by
\begin{equation}
    \Delta^\mathrm{aux}\colon f\mapsto \frac{1}{2}\mathrm{div}_{\mu^\mathrm{aux}}\{f,\enspace\}_{\omega^\mathrm{aux}}
\end{equation}
satisfies $(\Delta^\mathrm{aux})^2=0$. Moreover, the action functional $\calS^\mathrm{aux}$ satisfies 
\begin{equation}
    Q(\calS^\mathrm{aux})+\frac{1}{2}\{\calS^\mathrm{aux},\calS^\mathrm{aux}\}_{\omega^\mathrm{aux}}-\I\hbar \Delta^\mathrm{aux}\calS^\mathrm{aux}=0,
\end{equation}
which is equivalent to 
\begin{equation}
    \delta_\mathrm{BV}^\mathrm{aux}\ee^{\frac{\I}{\hbar}\calS^\mathrm{aux}}:=(Q-\I\hbar\Delta^\mathrm{aux})\ee^{\frac{\I}{\hbar}\calS^\mathrm{aux}}=0.
\end{equation}
\end{defn}
\begin{rem}
    The name ``semi-quantum'' is chosen since it is not a quantum observable yet, but rather the theory whose functional integral quantization will lead to a quantum observable in the sense that it is closed with respect to the infinitesimal symmetries.
\end{rem}

We also want to extend the notion of equivalent pre-observables to the case of semi-quantum pre-observables.

\begin{defn}[Equivalent BV semi-quantum pre-observables]
Two BV semi-quantum pre-observables 
\[
(\calF^\mathrm{aux},\calS^\mathrm{aux},\omega^\mathrm{aux},\mu^\mathrm{aux})\quad \text{and}\quad (\calF^\mathrm{aux},\Tilde{\calS}^\mathrm{aux},\omega^\mathrm{aux},\mu^\mathrm{aux}) 
\]
are said to be \emph{equivalent} if there exists a function $f^\mathrm{aux}\in C^\infty(\calF\times \calF^\mathrm{aux})$ such that 
\begin{equation}
    \ee^{\frac{\I}{\hbar}\Tilde{\calS}^\mathrm{aux}}-\ee^{\frac{\I}{\hbar}\calS^\mathrm{aux}}=(Q-\I\hbar\Delta^\mathrm{aux})\left(\ee^{\frac{\I}{\hbar}\calS^\mathrm{aux}}f^\mathrm{aux}\right).
\end{equation}
\end{defn}

\begin{prop}[\cite{CattRoss2005,Mn3}]
\label{prop:from_pre-observable_to_observable}
Let $(\calF^\mathrm{aux},\calS^\mathrm{aux},\omega^\mathrm{aux},\mu^\mathrm{aux})$ be a BV semi-quantum pre-observable. Define 
\begin{equation}
\label{eq:observable_path_integral}
    \calO_\calL:=\int_{\calL\subset\calF^\mathrm{aux}}\ee^{\frac{\I}{\hbar}\calS^\mathrm{aux}}\sqrt{\mu^\mathrm{aux}}\vert_\calL\in C^\infty(\calF),
\end{equation}
where $\calL\subset\calF^\mathrm{aux}$ is a Lagrangian submanifold. Then $\calO_\calL$ is an observable, i.e. $Q(\calO_\calL)=0$. Moreover, if for two Lagrangian submanifolds $\calL$ and $\Tilde{\calL}$ there exists a homotopy between them, then the observables $\calO_\calL$ and $\calO_{\Tilde{\calL}}$ are equivalent. Also for two equivalent BV semi-quantum pre-observables $\calS^\mathrm{aux}$ and $\Tilde{\calS}^\mathrm{aux}$, the corresponding observables $\calO_\calL$ and $\Tilde{\calO}_\calL$ are equivalent.
\end{prop}

\begin{defn}[Good auxiliary splitting]
We say that a semi-quantum pre-observable $(\calF^\mathrm{aux},\calS^\mathrm{aux},\omega^\mathrm{aux},\mu^\mathrm{aux})$ has a \emph{good splitting} if there is a decomposition $$\calF^\mathrm{aux}=\mathsf{F}^\mathrm{aux}\times \mathscr{F}^\mathrm{aux}$$ such that 
\begin{align}
    \omega^\mathrm{aux}&=\omega_1^\mathrm{aux}+\omega_2^\mathrm{aux},\\
    \mu^\mathrm{aux}&=\mu^\mathrm{aux}_1\otimes\mu^\mathrm{aux}_2,
\end{align}
where $\omega_1^\mathrm{aux}$ is a symplectic form on $\mathsf{F}^\mathrm{aux}$, $\omega_2^\mathrm{aux}$ is a symplectic form on $\mathscr{F}^\mathrm{aux}$, $\mu^\mathrm{aux}_1$ is a volume form on $\mathsf{F}^\mathrm{aux}$ and $\mu^\mathrm{aux}_2$ is a volume form on $\mathscr{F}^\mathrm{aux}$. 
\end{defn}

\begin{rem}
This is in fact the trivial case. The general version, called \emph{hedgehog}, is discussed in \cite{CMR2}.
\end{rem}

\begin{rem}
We split the auxiliary fields into \emph{high energy modes} $\mathscr{F}^\mathrm{aux}$ and \emph{low energy modes} $\mathsf{F}^\mathrm{aux}$. This splitting can be done by using \emph{Hodge decomposition} of differential forms into exact, coexact and harmonic forms (see Appendix A of \cite{CMR2}). Note that, in addition, we might also have \emph{background fields}\footnote{These
are background choices for classical fields that are not fixed by the
boundary conditions and the Euler–Lagrange equations.}.
If $\Sigma_d$ would have boundary, one can in general split the space of fields into three parts, the low energy fields, the high energy fields and the boundary fields. The boundary fields are generally given by techniques of symplectic reduction as the leaves of a chosen polarization on the boundary. This is the content of the BV-BFV formalism \cite{CMR1,CMR2,CattMosh1}. 
\end{rem}

\begin{prop}[\cite{Mn3}]
Let $(\calF^\mathrm{aux},\calS^\mathrm{aux},\omega^\mathrm{aux},\mu^\mathrm{aux})$ be a semi-quantum pre-observable with a good splitting. Define $\mathsf{S}^\mathrm{aux}\in C^\infty(\calF\times \mathsf{F}^\mathrm{aux})$ by
\begin{equation}
\label{eq:low_energy_effective}
    \mathsf{S}^\mathrm{aux}=-\I\hbar\log\int_{\mathscr{L}\subset\mathscr{F}^\mathrm{aux}}\ee^{\frac{\I}{\hbar}\calS^\mathrm{aux}}\sqrt{\mu_2^\mathrm{aux}}\big|_{\mathscr{L}},
\end{equation}
where $\mathscr{L}$ is a Lagrangian submanifold of $\mathscr{F}^\mathrm{aux}$. Then $(\mathsf{F}^\mathrm{aux},\mathsf{S}^\mathrm{aux},\omega^\mathrm{aux}_1,\mu^\mathrm{aux}_1)$ defines a semi-quantum pre-observable for the same BV theory. Moreover, the observable for the BV theory induced by $\mathsf{S}^\mathrm{aux}$ using Equation \eqref{eq:observable_path_integral} with a Lagrangian submanifold $\mathsf{L}\subset \mathsf{F}^\mathrm{aux}$ is equivalent to the one induced by $\calS^\mathrm{aux}$ using the Lagrangian submanifold $\calL \subset \calF^\mathrm{aux}$, if there exists a homotopy between $\calL$ and $\mathsf{L}\times \mathscr{L}$ in $\calF^\mathrm{aux}$.
\end{prop}
\begin{rem}
    Note that Equation \eqref{eq:low_energy_effective} means that $\mathsf{S}^\mathrm{aux}$ is the \emph{low energy effective action} (zero modes).
\end{rem}

\subsection{Observables for quantum BV manifolds}

\begin{defn}[BV quantum observable]
A \emph{BV quantum observable} for a quantum BV manifold is a function $\calO$ on $\calF$ of degree 0 such that 
\begin{equation}
    \delta_{\mathrm{BV}}\calO=0\Longleftrightarrow \Delta\left(\calO\ee^{\frac{\I}{\hbar}\calS}\right)=0.
\end{equation}
\end{defn}
\begin{defn}[Equivalent BV quantum observables]
Two BV quantum observables $\calO$ and $\Tilde{\calO}$ are said to be equivalent if 
\begin{equation}
    \Tilde{\calO}-\calO=\delta_\mathrm{BV}\Psi,\qquad \Psi\in C^\infty(\calF),
\end{equation}
or equivalently, $\calO$ and $\Tilde{\calO}$ have the same $\delta_\mathrm{BV}$-cohomology class.
\end{defn}

\begin{defn}[BV quantum pre-observable]
A \emph{BV quantum pre-observable} for a BV manifold is a BV semi-quantum pre-observable $$(\calF^\mathrm{aux},\omega^\mathrm{aux},\mu^\mathrm{aux},\calS^\mathrm{aux})$$ where $\calS+ \calS^\mathrm{aux}$ satisfies the QME
\begin{equation}
    (\Delta+\Delta^\mathrm{aux})\ee^{\frac{\I}{\hbar}(\calS+\calS^\mathrm{aux})}=0.
\end{equation}
\end{defn}

\begin{prop}[\cite{Mn3}]
Let $(\calF^\mathrm{aux},\omega^\mathrm{aux},\mu^\mathrm{aux},\calS^\mathrm{aux})$ be a BV quantum pre-observable. Define
\begin{equation}
    \calO_{\calL}:=\int_{\calL\subset \calF^\mathrm{aux}}\ee^{\frac{\I}{\hbar}\calS^\mathrm{aux}}\sqrt{\mu^\mathrm{aux}}\vert_\calL\in C^\infty(\calF),
\end{equation}
where $\calL\subset \calF^\mathrm{aux}$ is a Lagrangian submanifold. Then $\calO_\calL$ is an observable, i.e. $\delta_\mathrm{BV}\calO_\calL=0$. Moreover, if for two Lagrangian submanifolds $\calL$ and $\Tilde{\calL}$ there exists a homotopy between them, then the observables $\calO_\calL$ and $\calO_{\Tilde{\calL}}$ are equivalent.
\end{prop}

\section{Formal global split AKSZ sigma models}
\label{sec:formal_global_split_AKSZ_sigma_models}
The formal global construction for ASKZ sigma models is given by using methods of formal geometry (see \cite{GK,B} for the formal geometry part, and \cite{CMW4} for a detailed discussion of the formal global split AKSZ construction and its quantization) where one constructs a BV action that depends on a choice of classical background by adding an additional term to the AKSZ-BV action. This construction leads to modifications in the usual BV gauge-fixing condition if we apply the BV construction to this new formal global action. The globalization arises in an equivalent way as for the constructions involving the underlying curved\footnote{An $L_\infty$-algebra $\mathfrak{g}$ is called \emph{curved} if there exists an operation $\ell_0\colon \R\to \mathfrak{g}$ of degree 0. In particular, the strong homotopy Jacobi identity implies that $\ell_1\circ \ell_1=\pm \ell_2(\ell_0,\enspace)$, meaning that the unary bracket $\ell_1$ does not square to zero anymore, as it is the case for usual $L_\infty$-algebras. In this case we say that $\ell_1$ has non-vanishing curvature, thus the name ``curved''.} $L_\infty$-structure for the space of fields (see e.g. \cite{LazarevSchedler2012} for an exposition on curved $\infty$-structures and \cite{Cost1} for the field theoretic concept).   

In this section we want to recall some notions of formal geometry and describe the extension of AKSZ sigma models to a formal global version. 

\subsection{Notions of formal geometry}
\label{subsec:notions_of_formal_geometry}
Let us introduce the main players.

\label{sec:formal_geometry}
\begin{defn}[Generalized exponential map]
Let $M$ be a manifold and let $U\subset TM$ be an open neighborhood of the zero section of the tangent bundle. A \emph{generalized exponential map} is a map $\phi\colon U\to M$ such that $\phi\colon (x,p)\mapsto \phi_x(p)$ with $\phi_x(0)=x$ and $\dd\phi_x(0)=\id_{T_xM}$. Locally, we have 
\begin{equation}
\label{eq:generalized_exponential_map_local_coordinates}
    \phi^{i}_x(p)=x^{i}+p^{i}+\frac{1}{2}\phi^{i}_{x,jk}p^jp^k+\frac{1}{3!}\phi^{i}_{x,jk\ell}p^jp^kp^\ell+\dotsm
\end{equation}
where $(x^{i})$ are coordinates on the base and $(p^{i})$ are coordinates on the fiber.
\end{defn}

\begin{defn}[Formal exponential map]
A \emph{formal exponential map} is an equivalence class of generalized exponential maps, where we identify two generalized exponential maps if their jets agree to all orders. 
\end{defn}

One can define a flat connection $D$ on $\Hat{\Sym}(T^*M)$, where $\Hat{\Sym}$ denotes the completed symmetric algebra. Such a flat connection $D$ is called \emph{classical} \emph{Grothendieck} connection \cite{CF3} and it is locally given by $D=\dd_M+R$, where $$R\in \Omega^1\left(M,\mathrm{Der}\left(\Hat{\Sym}(T^*M)\right)\right)$$ is a 1-form with values in derivations of the completed symmetric algebra of the cotangent bundle. Here $R$ acts on sections $\sigma\in \Gamma\left(\Hat{\Sym}(T^*M)\right)$ by Lie derivative, that is $R(\sigma)=L_R\sigma$. Note that we have denoted by $\dd_M$ the de Rham differential on $M$. In local coordinates we have $R=R_\ell\dd_M x^{\ell}$, where $R_\ell=R_\ell^j(x,p)\frac{\de}{\de p^j}$ and 
\begin{equation}
    R_\ell^j(x,p)=-\frac{\de \phi^k}{\de x^\ell}\left(\left(\frac{\de \phi}{\de p}\right)^{-1}\right)^j_k=-\delta_\ell^j+O(p).
\end{equation}
Hence, for $\sigma\in \Gamma\left(\Hat{\Sym}(T^*M)\right)$ we have
\begin{equation}
\label{eq:R_vector_field}
    R(\sigma):=L_R(\sigma)=R_\ell(\sigma)\dd_M x^\ell=-\frac{\de\sigma}{\de p^j}\frac{\de \phi^k}{\de x^\ell}\left(\left(\frac{\de \phi}{\de p}\right)^{-1}\right)^j_k\dd_M x^\ell.
\end{equation}

Note that we can extend the connection $D$ to the complex $$\Gamma\left(\bigwedge^\bullet T^*M\otimes \Hat{\Sym}(T^*M)\right)$$ of $\Hat{\Sym}(T^*M)$-valued differential forms\footnote{Since $\Gamma\left(\bigwedge^\bullet T^*M\otimes \Hat{\Sym}(T^*M)\right)$ is the algebra of functions on the formal graded manifold $T[1]M\oplus T[0]M$, the differential $D$ turns this graded manifold into a differential graded manifold. In particular, since $D$ vanishes on the body of the graded manifold, we can linearize at each $x\in M$ and obtain an $L_\infty$-structure on $T_xM[1]\oplus T_xM$.}.
The following proposition tells us that the $D$-closed sections are exactly given by smooth functions. 
\begin{prop}
\label{prop:zero_cohomology}
A section $\sigma\in \Gamma\left(\Hat{\Sym}(T^*M)\right)$ is $D$-closed if and only if $\sigma=\mathsf{T}\phi^*f$ for some $f\in C^\infty(M)$, where $\mathsf{T}$ denotes the Taylor expansion around the fiber coordinates at zero. Moreover, the $D$-cohomology $$H^\bullet_D\left(\Hat{\Sym}(T^*M)\right)$$ is concentrated in degree 0 and
\begin{equation}
    H^0_D\left(\Hat{\Sym}(T^*M)\right)=\mathsf{T}\phi^*C^\infty(M)\cong C^\infty(M).
\end{equation}
\end{prop}
\begin{rem}
Note that we use any representative of $\phi$ to define the pullback.
\end{rem}

\begin{proof}[Proof of Proposition \ref{prop:zero_cohomology}]
If we use \eqref{eq:generalized_exponential_map_local_coordinates} and \eqref{eq:R_vector_field}, We can see that $R=\delta+R'$ where $\delta=\dd x^{i}\frac{\de}{\de p^{i}}$ and $R'$ is a 1-form with values in vector fields vanishing at $p=0$. Then we have $D=\delta+D'$ with \begin{equation}
    D'=\dd x^{i}\frac{\de}{\de x^{i}}+R'.
\end{equation}
One should note that $\delta$ is itself a differential and that it decreases the polynomial degree in $p$, whereas $D'$ does not decrease the degree. We can show that the cohomology of $\delta$ consists of 0-forms which are constant in $p$. To show this, let $$\delta^*=p^{i}\iota_{\frac{\de}{\de x^{i}}}$$ and note that 
\begin{equation}
    (\delta\delta^*+\delta^*\delta)\sigma=k\sigma,
\end{equation}
where $\sigma$ is an $r$-form of degree $s$ in $p$ such that $r+s=k$. By cohomological perturbation theory the cohomology of $D$ is isomorphic to the cohomology of $\delta$. 
\end{proof}
Note that in local coordinates we get for $f\in C^\infty(M)$ 
\begin{equation}
    \T\phi_x^*f=f(x)+p^{i}\de_if(x)+\frac{1}{2}p^jp^k(\de_j\de_kf(x)+\phi^{i}_{x,jk}\de_if(x))+\dotsm 
\end{equation}

An interesting question is how the Grothendieck connection depends on the choice of formal exponential map. Let $I\subset\R$ be an open interval and let $\phi$ be a family of formal exponential maps depending on a parameter $t\in I$. This family may be associated to a family of formal exponential maps $\psi$ on $M\times I$ by 
\begin{equation}
    \psi(x,t,p,\tau)=(\phi_{x,t}(p),t+\tau),
\end{equation}
where $\tau$ denotes the tangent coordinate to $t$. The associated connection $\Tilde{R}$ is defined by 
\begin{equation}
    \Tilde{R}\left(\Tilde{\sigma}\right)=-(\dd_p\Tilde{\sigma},\dd_\tau\Tilde{\sigma})\circ\begin{pmatrix}(\dd_p\phi)^{-1}&0\\0&1\end{pmatrix}\circ \begin{pmatrix}\dd_x\phi&\dot{\phi}\\0&1\end{pmatrix},\quad \Tilde{\sigma}\in \Gamma\left(\Hat{\Sym}(T^*(M\times I)\right).
\end{equation}
Thus we can write $\Tilde{R}=R+C\dd t+T$ with $R$ defined as before with the difference that it now depends on $t$, $C$ is given by 
\begin{equation}
C(\Tilde{\sigma})=-\dd_p\Tilde{\sigma}\circ (\dd_p\phi)^{-1}\circ \dot{\phi},
\end{equation}
and $T=-\dd t\frac{\de}{\de \tau}$. Note that $\dd_xT=0$, $\dd_tT=0$ and $[T,R]=0$, $[T,C]=0$. Thus, using the Maurer--Cartan equation for $\Tilde{R}$ and for $R$, we get 
\begin{equation}
    \dot{R}=\dd_xC+[R,C],
\end{equation}
which shows that under a change of formal exponential map, $R$ changes by a gauge transformation with generator $C$. Moreover, if $\sigma=\T\phi_x^*f$ for some $f\in C^\infty(M\times I)$, we get
\begin{equation}
    \dot{\sigma}=-L_C\sigma.
\end{equation}
This can be thought of as an associated gauge transformation for sections.

\subsection{Formal global AKSZ sigma models}
\label{sec:formal_global_AKSZ_models}
Let $\Sigma_d$ be a closed, oriented, compact $d$-manifold and consider a Hamiltonian $Q$-manifold $$(\calM,\omega_\calM=\dd_\calM\alpha_\calM,\Theta_\calM,Q_\calM)$$
of degree $d-1$. As described in Section \ref{sec:AKSZ_sigma_models}, we can consider its induced AKSZ theory with the space of fields 
\begin{equation}
    \calF^\calM_{\Sigma_d}=\Map_{\mathrm{GrMnf}}(T[1]\Sigma_d,\calM).
\end{equation}
Consider now a formal exponential map $\phi\colon T\calM\to \calM$. Then we can lift the space of fields by $\phi$. For $x\in \calM$ we denote the lifted space of fields by 
\begin{equation}
    \Hat{\calF}^\calM_{\Sigma_d}:=\Map_{\mathrm{GrMnf}}(T[1]\Sigma_d,T_x\calM)\cong \Omega^\bullet(\Sigma_d)\otimes T_x\calM.
\end{equation}
Note that we have used the fact that 
\begin{equation}
\label{eq:iso}
C^\infty(T[1]\Sigma_d)\cong \Omega^\bullet(\Sigma_d).
\end{equation}
This construction gives us a linear space for the target and thus we can identify the fields with differential forms on $\Sigma_d$ with values in the vector space $T_x\calM$ for $x\in \calM$. 
Consider the map 
\begin{equation}
    \Tilde{\phi}_x\colon \Hat{\calF}^\calM_{\Sigma_d,x}\to \calF^\calM_{\Sigma_d},
\end{equation}
which is given by composition with $\phi^{-1}_x$, i.e. $\Hat{\calF}^\calM_{\Sigma_d,x}=\phi^{-1}_x\circ \calF^\calM_{\Sigma_d}$.
We can lift the BV symplectic 2-form $\omega_{\Sigma_d}$, the primitive 1-form $\alpha_{\Sigma_d}$ and the BV action $\calS_{\Sigma_d}$ to the lifted space of fields. We will denote the lifts by 
\begin{align}
    \Hat{\alpha}_{\Sigma_d,x}&=\Tilde{\phi}^*_x\alpha_{\Sigma_d}\in \Omega^1\left(\Hat{\calF}^\calM_{\Sigma_d,x}\right),\\
    \Hat{\omega}_{\Sigma_d,x}&=\Tilde{\phi}^*_x\omega_{\Sigma_d}\in \Omega^2\left(\Hat{\calF}^\calM_{\Sigma_d,x}\right),\\
    \Hat{\calS}^{\mathrm{AKSZ}}_{\Sigma_d,x}&=\iota_{\Hat{\dd}_{\Sigma_d}}\Tilde{\phi}^*_x\mathscr{T}_{\Sigma_d}(\alpha_\calM)+\mathsf{T}\Tilde{\phi}^*_x\mathscr{T}_{\Sigma_d}(\Theta_\calM)\in C^\infty\left(\Hat{\calF}^\calM_{\Sigma_d,x}\right).
\end{align}

Note that we can regard a constant map $x\colon T[1]\Sigma_d\to \calM$ in $\calF^\calM_{\Sigma_d}$ as an element of $\calM$, hence there is a natural inclusion $\calM\hookrightarrow \calF^\calM_{\Sigma_d}$.
For a constant field $x$ and $\calA\in\calF^\calM_{\Sigma_d}$ We can construct a 1-form 
\begin{equation}
R_{\Sigma_d}=(R_{\Sigma_d})_\mu(x,\calA)\dd_\calM x^\mu
\end{equation}
on $\calM$ with values in differential operators on $\calF_{\Sigma_d}^\calM$. Moreover, we can lift this 1-form to $\Hat{\calF}_{\Sigma_d}^\calM$ and we denote the lift by $\Hat{R}_{\Sigma_d}$. Locally, we write
\begin{equation}
\Hat{R}_{\Sigma_d}=\left(\Hat{R}_{\Sigma_d}\right)_\mu\left(x,\Hat{\calA}\right)\dd_\calM x^\mu. 
\end{equation}
It is important to recall that classical solutions for AKSZ sigma models, i.e. solutions of $\delta\calS_{\Sigma_d}=0$, are given by differential graded maps $$(T[1]\Sigma_d,\dd_{\Sigma_d})\to (\calM,Q_\calM).$$ 
Hence we can consider the moduli space of classical solutions $\mathrm{M}_{\mathrm{cl}}$ for AKSZ theories which is given by constant maps $x\colon T[1]\Sigma_d\to \calM$ and thus we get an isomorphism $\mathrm{M}_{\mathrm{cl}}\cong\calM$. We will refer to this constant solutions as being \emph{background fields}. Choosing a background field $x\in\calM$, we can define a formal global AKSZ action.
\begin{defn}[Formal global AKSZ action]
The \emph{formal global AKSZ action} is given by
\begin{equation}
    \Hat{\calS}^\mathrm{global}_{\Sigma_d,x}=\iota_{\Hat{\dd}_{\Sigma_d}}\Tilde{\phi}_x^*\mathscr{T}_{\Sigma_d}(\alpha_\calM)+\T\Tilde{\phi}^*_x\mathscr{T}_{\Sigma_d}(\Theta_\calM)+\Hat{\calS}_{\Sigma_d,R,x},
\end{equation}
where $\Hat{\calS}_{\Sigma_d,R,x}$ is constructed locally such that
\begin{equation}
    \Hat{\calS}_{\Sigma_d,R,x}\left(\Hat{\calA}\right)=\int_{\Sigma_d}\left(\Hat{R}_{\Sigma_d}\right)_\mu\left(x,\Hat{\calA}\right)\dd_\calM x^\mu.
\end{equation}
\end{defn}
Hence locally we get we get 
\begin{equation}
    \Hat{\calS}^\mathrm{global}_{\Sigma_d,x}=\int_{\Sigma_d}\Hat{\alpha}_{\mu}\left(\Hat{\calA}\right)\dd_{\Sigma_d}\Hat{\calA}^\mu+\int_{\Sigma_d}\Hat{\Theta}_{\calM,x}\left(\Hat{\calA}\right)+\int_{\Sigma_d}\left(\Hat{R}_{\Sigma_d}\right)_\mu\left(x,\Hat{\calA}\right)\dd_\calM x^\mu,
\end{equation}
where $\Hat{\alpha}_\mu$ are the coefficients of $\Hat{\alpha}_{\Sigma_d,x}:=\Tilde{\phi}^*_x\alpha_{\Sigma_d}$ and $\Hat{\Theta}_{\calM,x}:=\mathsf{T}\Tilde{\phi}_x^*\Theta_{\calM}$.

\begin{rem}
This construction has to be understood in a formal way. The geometric meaning and the relation to a global construction is clear when using the relation of $R_{\Sigma_d}$ to the Grothendieck connection $D$. This can be done if we start with a theory called \emph{split} which we will introduce now.
\end{rem}

\subsection{Formal global split AKSZ sigma models}
\label{sec:formal_global_split_AKSZ_construction}
AKSZ theories can generally be more difficult to work with depending on the target differential graded symplectic manifold $\calM$. Recall that, using the isomorphism \eqref{eq:iso}, if the target is linear, we have an isomorphism 
\begin{equation}
\label{eq:isomorphism_1}
    \Map_{\mathrm{GrMnf}}(T[1]\Sigma_d,\calM)\cong \Omega^\bullet(\Sigma_d)\otimes \calM.
\end{equation}
Moreover, we can split the space of fields by considering $\calM$ to be the shifted cotangent bundle of a linear space. At first, however, we only want $\calM$ to be the shifted cotangent bundle of any graded manifold $M$.
This leads to the following definition of AKSZ theories.

\begin{defn}[Linear split AKSZ sigma model]
\label{defn:split_like_AKSZ_sigma_model}
We call a $d$-dimensional AKSZ sigma model \emph{linear split} if the target is of the form $$\calM=V\oplus V^*$$ for some vector space $V$.
\end{defn}

\begin{defn}[Split AKSZ sigma model]
\label{defn:split_AKSZ_sigma_model}
We call a $d$-dimensional AKSZ sigma model \emph{split} if the target is of the form $$\calM=T^*[d-1]M$$ for some graded manifold $M$.
\end{defn}

This space can be lifted to a formal construction using methods of formal geometry as in Section \ref{sec:formal_geometry} to the shifted cotangent bundle of the tangent space of $M$ at some constant background in $M$.
Consider a $d$-dimensional split AKSZ sigma model with space of fields given by 
\begin{equation}
\label{eq:BV_space_of_fields_split}
    \calF^\calM_{\Sigma_d}=\Map_{\mathrm{GrMnf}}(T[1]\Sigma_d,T^*[d-1]M),
\end{equation}
for some graded manifold $M$, with its corresponding AKSZ-BV theory $$\left(\calF^\calM_{\Sigma_d},\calS_{\Sigma_d},\omega_{\Sigma_d}\right).$$ 
Note that, similarly as for general AKSZ theories, one type of classical solutions to the Euler--Lagrange equations for split AKSZ theories are given by fields of the form $(x,0)$ where $x\colon \Sigma_d\to M$ is a constant background field. Note that the classical space of fields $F_{\Sigma_d}$ is given by vector bundle maps $T\Sigma_d\to T^*M$, i.e. \[F_{\Sigma_d}=\Map_{\mathrm{VecBun}}(T\Sigma_d,T^*M).\] 
Then the BV space of fields is given by \eqref{eq:BV_space_of_fields_split}.
Thus, for the classical space of fields $F_{\Sigma_d}$, we have a moduli space of classical solutions
\begin{equation}
    \mathrm{M}_{\mathrm{cl}}=\left\{(A,B)\in \Map(T\Sigma_d,T^*M)\mid A=x=const,B=0\right\}\cong M.
\end{equation}
Moreover, for a chosen formal exponential map $\phi\colon TM\to M$ and a constant background field $x\colon \Sigma_d\to M$ regarded as an element of the moduli space of classical solutions $\mathrm{M}_{\mathrm{cl}}$, one can consider the lifted space of fields 
\begin{equation}
    \Hat{\calF}^\calM_{\Sigma_d,x}=\Map_{\mathrm{GrMnf}}(T[1]\Sigma_d,T^*[d-1]T_xM),
\end{equation}
which gives a \emph{linearization} (or also \emph{coordinatization}) of the space of fields in the target as we have seen before. 
Let $(\boldsymbol{A},\boldsymbol{B})\in \calF^\calM_{\Sigma_d}$, where $\boldsymbol{A}\colon T[1]\Sigma_d\to M$ denotes the base superfield and $\boldsymbol{B}\in \Gamma(\Sigma_d,T^*\Sigma_d\otimes\boldsymbol{A}^*T^*[d-1]M)$ the fiber superfield.
Consider the corresponding lifts by $\phi$ where the superfields are given by 
\begin{equation}
    \Hat{\boldsymbol{A}}:=\phi_x^{-1}(\boldsymbol{A}),\qquad \Hat{\boldsymbol{B}}:=(\dd\phi_x)^{*}\boldsymbol{B}
\end{equation}
The BV action functional $\calS_{\Sigma_d}$ then lifts to a \emph{formal global action}.
\begin{defn}[Formal global split AKSZ action]
The \emph{formal global action} for the \emph{split AKSZ sigma model} is given by
\begin{equation}
\boxed{
    \Hat{\calS}^\mathrm{global}_{\Sigma_d,x}:=\int_{\Sigma_d}\Hat{\boldsymbol{B}}_\ell\land\dd_{\Sigma_d}\Hat{\boldsymbol{A}}^\ell+\int_{\Sigma_d}\Hat{\Theta}_{\calM,x}\left(\Hat{\boldsymbol{A}},\Hat{\boldsymbol{B}}\right)+\int_{\Sigma_d}R_\ell^j\left(x,\Hat{\boldsymbol{A}}\right)\Hat{\boldsymbol{B}}_j\land\dd_{M} x^\ell.
    }
\end{equation}
\end{defn}
\begin{rem}
Note that in this case we get a lift of $R$ as defined in Section \ref{sec:formal_geometry} to the space of fields which splits into base and fiber fields by  
\begin{equation}
\Hat{\calF}^\calM_{\Sigma_d,x}\cong \Omega^\bullet(\Sigma_d)\otimes T_xM\oplus \Omega^\bullet(\Sigma_d)\otimes T^*_xM[d-1].
\end{equation}
Hence the induced 1-form $\Hat{R}_{\Sigma_d}$ is indeed given by 
\begin{equation}
    \Hat{R}_{\Sigma_d}=R_\ell^j\left(x,\Hat{\boldsymbol{A}}\right)\Hat{\boldsymbol{B}}_j\land \dd_M x^\ell,
\end{equation}
where $R_\ell^j$ are the components of $R\in \Omega^1\left(M,\Der\left(\Hat{\Sym}(T^*M)\right)\right)$.
\end{rem}
The $Q$-structure is given by the Hamiltonian vector field of $\Hat{\calS}^\mathrm{global}_{\Sigma_d,x}$. Indeed, let $\Hat{R}_{\Sigma_d}$ denote the lift of the vector field $R_{\Sigma_d}$ to $\Hat{\calF}^\calM_{\Sigma_d,x}$ and let 
\begin{align}
    \Hat{\calS}^\mathrm{AKSZ}_{\Sigma_d,x}&:=\int_{\Sigma_d}\Hat{\boldsymbol{B}}_\ell\land\dd_{\Sigma_d}\Hat{\boldsymbol{A}}^\ell+\int_{\Sigma_d}\Hat{\Theta}_{\calM,x}\left(\Hat{\boldsymbol{A}},\Hat{\boldsymbol{B}}\right)\\
    \Hat{\calS}_{\Sigma_d,R,x}&:=\int_{\Sigma_d}R_\ell^j\left(x,\Hat{\boldsymbol{A}}\right)\Hat{\boldsymbol{B}}_j\land\dd_M x^\ell,
\end{align}
such that 
\begin{equation}
\label{eq:decomposition_global_action}
    \Hat{\calS}^\mathrm{global}_{\Sigma_d,x}=\Hat{\calS}^\mathrm{AKSZ}_{\Sigma_d,x}+\Hat{\calS}_{\Sigma_d,R,x}.
\end{equation}
Denote by $\Hat{\omega}_{\Sigma_d,x}=\Tilde{\phi}_x^*\omega_{\Sigma_d}$ the lift of the symplectic form on $\calF^\calM_{\Sigma_d}$ to a symplectic form on $\Hat{\calF}^\calM_{\Sigma_d,x}$.
Then we can define a cohomological vector field $\Hat{Q}_{\Sigma_d,x}$ on $\Hat{\calF}^\calM_{\Sigma_d,x}$ by
\begin{equation}
    \Hat{Q}_{\Sigma_d,x}=\Hat{Q}^\mathrm{AKSZ}_{\Sigma_d,x}+\Hat{R}_{\Sigma_d},
\end{equation}
where $\Hat{Q}^\mathrm{AKSZ}_{\Sigma_d,x}$ is the Hamiltonian vector field of 
\begin{equation}
    \mathrm{Back}_{\Hat{\calS}^\mathrm{AKSZ}_{\Sigma_d}}\colon x\mapsto\Hat{\calS}^\mathrm{AKSZ}_{\Sigma_d,x},
\end{equation}
and hence we have 
\begin{equation}
    \iota_{\Hat{Q}_{\Sigma_d,x}}\Hat{\omega}_{\Sigma_d,x}=\delta \Hat{\calS}^\mathrm{global}_{\Sigma_d,x}.
\end{equation}
This is in fact true if the source manifold is closed, i.e. $\de\Sigma_d=\varnothing$. We have denoted the map by ``Back'' to indicate the variation of the ``background''.
\begin{prop}
If $\de\Sigma_d=\varnothing$, then 
\begin{equation}
    \dd_x\mathrm{Back}_{\Hat{\calS}^\mathrm{AKSZ}_{\Sigma_d}}=\left\{\Hat{\calS}_{\Sigma_d,R,x},\mathrm{Back}_{\Hat{\calS}^\mathrm{AKSZ}_{\Sigma_d}}\right\}_{\Hat{\omega}_{\Sigma_d,x}},
\end{equation}
where $\dd_M$ denotes the de Rham differential on the moduli space space of classical solutions $\mathrm{M}_{\mathrm{cl}}\cong M$.
\end{prop}

Using the formal global action, we get the following Proposition (see also Proposition \ref{prop:dQME} for the quantum version)
\begin{prop}[dCME]
\label{prop:dCME}
The \emph{differential Classical Master Equation} for the formal global split AKSZ action holds:
\begin{equation}
\label{eq:dCME}
\boxed{
    \dd_x\Hat{\calS}^\mathrm{global}_{\Sigma_d,x}+\frac{1}{2}\left\{\Hat{\calS}^\mathrm{global}_{\Sigma_d,x},\Hat{\calS}^\mathrm{global}_{\Sigma_d,x}\right\}_{\Hat{\omega}_{\Sigma_d,x}}=0.}
\end{equation}
\end{prop}

\begin{defn}[Formal global split AKSZ sigma model]
The \emph{formal global split AKSZ sigma model} is given by the AKSZ-BV theory for the quadruple
\begin{equation}
\label{eq:formal_global_split_AKSZ_theory}
    \left(\Hat{\calF}^\calM_{\Sigma_d,x},\Hat{\calS}^\mathrm{global}_{\Sigma_d,x},\Hat{\omega}_{\Sigma_d,x},\Hat{Q}_{\Sigma_d,x}\right).
\end{equation}
\end{defn}
\begin{rem}
Note that the CME has to be replaced by the dCME as in \eqref{eq:dCME} in the formal global setting.
\end{rem}

\section{Pre-observables for AKSZ theories}

\subsection{AKSZ pre-observables}
Let $\Sigma_d$ be a closed and oriented source $d$-manifold and for some differential graded symplectic manifold $(\calN,\omega_\calN=\dd_\calN\alpha_\calN)$ let \[\pi\colon\calE=\calM\times\calN\to\calM\] 
be a trivial Hamiltonian $Q$-bundle of degree $n$ over some Hamiltonian $Q$-manifold $(\calM,\omega_\calM=\dd_\calM\alpha,Q_\calM,\Theta_\calM)$ of degree $d-1$. Denote by $\Theta_\calE\in C^\infty(\calE)$ the Hamiltonian on the total space $\calE$ and by $\calV_\calE\in \ker \dd\pi$ the vertical part of $Q_\calE$, such that \[Q_\calE=Q_\calM+\calV_\calE.\]
Consider the corresponding AKSZ-BV theory with BV manifold given by $$\left(\calF^\calM_{\Sigma_d},\calS_{\Sigma_d},\omega_{\Sigma_d},Q_{\Sigma_d}\right)$$ as it was constructed in Section \ref{sec:AKSZ_Theories}.
Let $i\colon\Sigma_k\hookrightarrow \Sigma_d$ be the embedding of a closed oriented submanifold of dimension $k\leq d$ and let the auxiliary space of fields be given by
\begin{equation}
    \calF^\calN_{\Sigma_k}:=\Map_{\mathrm{GrMnf}}(T[1]\Sigma_k,\calN).
\end{equation}
Moreover, consider the transgression maps 
\begin{align}
\mathscr{T}_{\Sigma_k}&\colon \Omega^\bullet(\calN)\to \Omega^\bullet\left(\calF^\calN_{\Sigma_k}\right),\\ \mathscr{T}^{\calE}_{\Sigma_k}&\colon \Omega^\bullet(\calE)\to \Omega^\bullet\left(\Map_{\mathrm{GrMnf}}(T[1]\Sigma_k,\calE)\right) 
\end{align}
corresponding to the fiber $\calN$ and the total space $\calE$.
Define\footnote{Note that extending to the case with auxiliary fields $\calF\times \calF^{\mathrm{aux}}$, we can extend $\pi$ to a map $\pi^\calN=\pi\times \id_{\calF^\calN_{\Sigma_k}}\colon \calF^\calM_{\Sigma_d}\times \calF^\calN_{\Sigma_k}\to \Map_{\mathrm{GrMnf}}(T[1]\Sigma_k,\calE)$.} 
$$p:=i^*\colon \underbrace{\Map_{\mathrm{GrMnf}}(T[1]\Sigma_d,\calM)}_{=:\calF^\calM_{\Sigma_d}}\to\underbrace{\Map_{\mathrm{GrMnf}}(T[1]\Sigma_k,\calM)}_{=:\calF^\calM_{\Sigma_k}}.$$
Furthermore, let $\Hat{\dd}_{\Sigma_k}\in\mathfrak{X}\left(\calF^\calN_{\Sigma_k}\right)\subset\mathfrak{X}^\mathrm{vert}\left(\calF^\calM_{\Sigma_d}\times\calF^\calN_{\Sigma_k}\right)$, where $\mathfrak{X}^\mathrm{vert}$ denotes the space of vertical vector fields, and let 
\[\Hat{\calV}_{\calE}\in \mathfrak{X}^\mathrm{vert}\left(\Map_{\mathrm{GrMnf}}(T[1]\Sigma_k,\calE)\right)\] 
be the lift of $\calV_{\calE}\in \mathfrak{X}^\mathrm{vert}(\calE)$ such that $p^*\Hat{\calV}_{\calE}\in \mathfrak{X}^\mathrm{vert}\left(\calF^\calM_{\Sigma_d}\times\calF^\calN_{\Sigma_k}\right)$, where $$p^*\colon C^\infty\left(\calF^\calM_{\Sigma_k}\right)\to C^\infty\left(\calF^\calM_{\Sigma_d}\right).$$
\begin{prop}[\cite{Mn3}]
Consider the data given by 
\begin{align}
    \calS^\calN_{\Sigma_k}&=\iota_{\Hat{\dd}_{\Sigma_k}}\mathscr{T}_{\Sigma_k}(\alpha_\calN)+p^*\mathscr{T}^{\calE}_{\Sigma_k}(\Theta_\calE),\\
    \omega^\calN_{\Sigma_k}&=(-1)^k\mathscr{T}_{\Sigma_k}(\omega_\calN),\\
    \calV^\calE_{\Sigma_k}&=\Hat{\dd}_{\Sigma_k}+p^*\Hat{\calV}_{\calE}.
\end{align}
Then the quadruple 
\begin{equation}
    \left(\calF^\calN_{\Sigma_k},\calS^\calN_{\Sigma_k},\omega^\calN_{\Sigma_k},\calV^\calE_{\Sigma_k}\right)
\end{equation}
defines a pre-observable for the AKSZ-BV theory as in \eqref{eq:formal_global_split_AKSZ_theory}, that is we have 
\begin{equation}
    Q_{\Sigma_d}\left(\calS^\calN_{\Sigma_k}\right)+\frac{1}{2}\left\{\calS^\calN_{\Sigma_k},\calS^\calN_{\Sigma_k}\right\}_{\omega^\calN_{\Sigma_k}}=0.
\end{equation}
\end{prop}
\begin{rem}
\label{rem:invariance_AKSZ_diffeo}
This pre-observable is invariant under reparamterizations of $\Sigma_k$ and under diffeomorphism of the ambient manifold $\Sigma_d$. In fact, for $(\calA,\calB)\in \calF^\calM_{\Sigma_d}\times \calF^\calN_{\Sigma_k}$, $\varphi_d\in \mathrm{Diff}(\Sigma_d)$ and $\varphi_k\in \mathrm{Diff}(\Sigma_k)$, one can immediately show that 
\begin{equation}
    \calS^\calN_{\Sigma_k}(\calA,\calB;\varphi_d\circ i\circ \varphi_k)=\calS^\calN_{\Sigma_k}\left(\varphi_d^*\calA,(\varphi_k)^{-1}\calB;i\right)
\end{equation}
\end{rem}

\subsection{Formal global AKSZ pre-observables}
We want to extend the constructions above to a formal global lift by using methods of formal geometry as in Section \ref{sec:formal_global_split_AKSZ_sigma_models}. It turns out that the formal global lift of the pre-observable constructed in the previous section is not automatically a pre-observable. In particular, it is spoilt by an obstruction which can be phrased as an equation that has to be satisfied. Hence we get the following theorem.
\begin{thm}
\label{thm:global_pre-observable}
Let $\left(\calF^\calM_{\Sigma_d},\calS_{\Sigma_d},\omega_{\Sigma_d},Q_{\Sigma_d}\right)$ be the AKSZ-BV theory constructed as before and let $i\colon \Sigma_k\hookrightarrow \Sigma_d$ be a submanifold of $\Sigma_d$. Moreover, consider constant background fields $x\in \calM$ and $y\in \calN$. 
Then its formal global AKSZ construction 
\begin{equation}
    \left(\Hat{\calF}^\calM_{\Sigma_d,x},\Hat{\calS}^\mathrm{global}_{\Sigma_d,x},\Hat{\omega}_{\Sigma_d,x},\Hat{Q}_{\Sigma_d,x}\right),
\end{equation}
constructed by using a formal exponential map $T\calM\to \calM$,
together with the formal global fiber 
\begin{equation}
\left(\Hat{\calF}^\calN_{\Sigma_k,y},\Hat{\calS}^\mathrm{global}_{\Sigma_k,y},\Hat{\omega}^\calN_{\Sigma_k,y}=\dd_\calN\Hat{\alpha}^\calN_{\Sigma_k,y}, \Hat{Q}^\calN_{\Sigma_k,y}\right),
\end{equation}
constructed using a formal exponential map $T\calN\to \calN$, defines a pre-observable if and only if
\begin{equation}
\label{eq:obstruction_1}
    \dd_y\Hat{\calS}^\mathrm{AKSZ}_{\Sigma_k,y}+\frac{1}{2}\left\{\Hat{\calS}_{\Sigma_k,R,y},\Hat{\calS}_{\Sigma_k,R,y}\right\}_{\Hat{\omega}_{\Sigma_k,y}}=0.
\end{equation}
\end{thm}
\begin{rem}
Moreover, for an exponential map $\phi\colon T\calN\to \calN$, we set
\begin{equation}
    \Hat{\calS}^\mathrm{AKSZ}_{\Sigma_k,y}=\underbrace{\Tilde{\phi}^*_y\iota_{\Hat{\dd}_{\Sigma_k}}\mathscr{T}_{\Sigma_k}(\alpha_\calN)}_{=:\Hat{\calS}^\mathrm{kin}_{\Sigma_k,y}}+\underbrace{\mathsf{T}\Tilde{\phi}_y^*p^*\mathscr{T}_{\Sigma_k}^\calE(\Theta_\calE)}_{=:\Hat{\calS}^\mathrm{target}_{\Sigma_k,y}},
\end{equation}
and thus we have a decomposition, similarly as in \eqref{eq:decomposition_global_action}, of the formal global action as 
\begin{equation}
    \Hat{\calS}^\mathrm{global}_{\Sigma_k,y}=\Hat{\calS}^\mathrm{AKSZ}_{\Sigma_k,y}+\Hat{\calS}_{\Sigma_k,R,y}.
\end{equation}
\end{rem}
The following Lemma is going to be useful for the proof of Theorem \ref{thm:global_pre-observable}.
\begin{lem}
\label{lem:Lemma1}
Let $\Sigma$ be a compact, connected manifold and let $\calM$ be a differential graded symplectic manifold. Moreover, let $X\in \mathfrak{X}(T[1]\Sigma)$, $Y\in \mathfrak{X}(\calM)$, $\Xi\in \Omega^\bullet(\calM)$ and denote the lifts of $X$ and $Y$ to the mapping space by $\Hat{X}$ and $\Hat{Y}$ respectively. Then 
\begin{align}
    L_{\Hat{X}}\mathscr{T}_\Sigma(\Xi)&=0,\\
    L_{\Hat{Y}}\mathscr{T}_\Sigma(\Xi)&=(-1)^{\gh(\Hat{Y})\dim\Sigma}\mathscr{T}_{\Sigma}(L_Y\Xi).
\end{align}
\end{lem}

\begin{proof}[Proof of Theorem \ref{thm:global_pre-observable}]
First we note that the lift 
\begin{equation}
\Hat{\calV}^\calE_{\Sigma_k,y}=\Hat{\dd}_{\Sigma_k}+\Tilde{\phi}^*_yp^*\Hat{\calV}_\calE
\end{equation}
of $\calV_{\Sigma_k}^\calE$
to $\mathfrak{X}^\mathrm{vert}\left(\Hat{\calF}^\calM_{\Sigma_d,x}\times \Hat{\calF}^\calN_{\Sigma_k,y}\right)$,
the space of vertical vector fields on the lifted mapping spaces,
is the Hamiltonian vector field for $\Hat{\calS}^\mathrm{AKSZ}_{\Sigma_k,y}$, i.e. we have 
\begin{equation}
    \Hat{\calV}^\calE_{\Sigma_k,y}=\left\{\Hat{\calS}^\mathrm{AKSZ}_{\Sigma_k,y},\enspace\right\}_{\Hat{\omega}^\calN_{\Sigma_k,y}}.
\end{equation}
Indeed, we have 
\begin{multline}
    \iota_{\Hat{\dd}_{\Sigma_k}}\Hat{\omega}_{\Sigma_k,y}=\iota_{\Hat{\dd}_{\Sigma_k}}(-1)^k\Tilde{\phi}^*_y\mathscr{T}_{\Sigma_k}(\omega_\calN)=\Tilde{\phi}_y^*\iota_{\Hat{\dd}_{\Sigma_k}}\delta\mathscr{T}_{\Sigma_k}(\alpha_\calN)\\=\Tilde{\phi}^*_y\underbrace{L_{\Hat{\dd}_{\Sigma_k}}\mathscr{T}_{\Sigma_k}(\alpha_\calN)}_{=0}+\Tilde{\phi}^*_y\delta\iota_{\Hat{\dd}_{\Sigma_k}}\mathscr{T}_{\Sigma_k}(\alpha_\calN)=\delta\Hat{\calS}^\mathrm{kin}_{\Sigma_k,y}=\delta^\mathrm{vert}\Hat{\calS}^\mathrm{kin}_{\Sigma_k,y},
\end{multline}
where we have used Cartan's magic formula $L=\dd\iota+\iota\dd$, Lemma \ref{lem:Lemma1} and the fact that $\Tilde{\phi}_y^*\Hat{\dd}_{\Sigma_k}=\Hat{\dd}_{\Sigma_k}$.
We have denoted by $\delta^\mathrm{vert}$ the vertical part of the de Rham differential $\delta$ on the lifted total mapping space $\Hat{\calF}^\calM_{\Sigma_d,x}\times \Hat{\calF}^\calN_{\Sigma_k,y}$, i.e. in the fiber direction $\Hat{\calF}^\calN_{\Sigma_k,y}$. The last equality holds since $\Hat{\calS}^\mathrm{kin}_{\Sigma_k,y}$ is constant in the $\Hat{\calF}^\calM_{\Sigma_d,x}$ direction. 
Similarly, we have 
\begin{multline}
    \iota_{\Tilde{\phi}^*_yp^*\Hat{\calV}_\calE}\Hat{\omega}_{\Sigma_k,y}=\mathsf{T}\Tilde{\phi}^*_yp^*\iota_{\Hat{\calV}_\calE}(-1)^k\mathscr{T}^\calE_{\Sigma_k}(\omega_\calN)\\=(-1)^k\mathsf{T}\Tilde{\phi}^*_yp^*\mathscr{T}^\calE_{\Sigma_k}(\underbrace{\iota_{\calV_\calE}\omega_\calN}_{=\delta^\mathrm{vert}\Theta_\calE})=\delta^\mathrm{vert}\mathsf{T}\Tilde{\phi}^*_yp^*\mathscr{T}^\calE_{\Sigma_k}(\Theta_\calE)=\delta^\mathrm{vert}\Hat{\calS}^\mathrm{target}_{\Sigma_k,y}.
\end{multline}
Moreover, we have 
\begin{multline}
\label{eq:number_1}
\Hat{Q}_{\Sigma_d,x}\left(\Hat{\calS}^\mathrm{global}_{\Sigma_k,y}\right)+\frac{1}{2}\left\{\Hat{\calS}^\mathrm{global}_{\Sigma_k,y},\Hat{\calS}^\mathrm{global}_{\Sigma_k,y}\right\}_{\Hat{\omega}^\calN_{\Sigma_k,y}}\\=\Hat{Q}_{\Sigma_d,x}\left(\Hat{\calS}^\mathrm{AKSZ}_{\Sigma_k,y}\right)+\Hat{Q}_{\Sigma_d,x}\left(\Hat{\calS}_{\Sigma_k,R,y}\right)+\frac{1}{2}\left\{\Hat{\calS}^\mathrm{AKSZ}_{\Sigma_k,y},\Hat{\calS}^\mathrm{AKSZ}_{\Sigma_k,y}\right\}_{\Hat{\omega}^\calN_{\Sigma_k,y}}\\+\underbrace{\left\{\Hat{\calS}_{\Sigma_k,R,y},\Hat{\calS}^\mathrm{AKSZ}_{\Sigma_k,y}\right\}_{\Hat{\omega}^\calN_{\Sigma_k,y}}}_{=\dd_y\Hat{\calS}^\mathrm{AKSZ}_{\Sigma_k,y}}+\frac{1}{2}\left\{\Hat{\calS}_{\Sigma_k,R,y},\Hat{\calS}_{\Sigma_k,R,y}\right\}_{\Hat{\omega}^\calN_{\Sigma_k,y}}.
\end{multline}
The first two terms of the left hand side of \eqref{eq:number_1} are given by
\begin{align}
\begin{split}
\label{eq:Q1}
    \Hat{Q}_{\Sigma_d,x}\left(\Hat{\calS}^\mathrm{AKSZ}_{\Sigma_k,y}\right)&=\Hat{Q}^\mathrm{AKSZ}_{\Sigma_d,x}\left(\Hat{\calS}^\mathrm{AKSZ}_{\Sigma_k,y}\right)+\Hat{R}_{\Sigma_d}\left(\Hat{\calS}^\mathrm{AKSZ}_{\Sigma_k,y}\right)\\&=\Hat{\dd}_{\Sigma_d}\left(\Hat{\calS}^\mathrm{AKSZ}_{\Sigma_k,y}\right)+\Tilde{\phi}_y^*\Hat{Q}_\calM\left(\Hat{\calS}^\mathrm{AKSZ}_{\Sigma_k,y}\right)+\Hat{R}_{\Sigma_d}\left(\Hat{\calS}^\mathrm{AKSZ}_{\Sigma_k,y}\right),
    \end{split}\\
    \label{eq:Q2}
    \begin{split}
    \Hat{Q}_{\Sigma_d,x}\left(\Hat{\calS}_{\Sigma_k,R,y}\right)&=\Hat{Q}^\mathrm{AKSZ}_{\Sigma_d,x}\left(\Hat{\calS}_{\Sigma_k,R,y}\right)+\Hat{R}_{\Sigma_d}\left(\Hat{\calS}_{\Sigma_k,R,y}\right)\\&=\Hat{\dd}_{\Sigma_d}\left(\Hat{\calS}_{\Sigma_k,R,y}\right)+\Tilde{\phi}_y^*\Hat{Q}_\calM\left(\Hat{\calS}_{\Sigma_k,R,y}\right)+\Hat{R}_{\Sigma_d}\left(\Hat{\calS}_{\Sigma_k,R,y}\right).
\end{split}
\end{align}
Since $\Hat{\calS}^\mathrm{kin}_{\Sigma_k,y}$ and $\Hat{\calS}_{\Sigma_k,R,y}$ are constant in direction of $\Hat{\calF}^\calM_{\Sigma_d,x}$, we get 
\begin{align}
\label{eq:Q3}
    \Hat{Q}_{\Sigma_d,x}\left(\Hat{\calS}^\mathrm{AKSZ}_{\Sigma_k,y}\right)&=\underbrace{\Hat{\dd}_{\Sigma_d}\left(\Hat{\calS}^\mathrm{target}_{\Sigma_k,y}\right)}_{=0}+\Tilde{\phi}_y^*\Hat{Q}_\calM\left(\Hat{\calS}^\mathrm{target}_{\Sigma_k,y}\right)+\Hat{R}_{\Sigma_d}\left(\Hat{\calS}^\mathrm{target}_{\Sigma_k,y}\right),\\
    \label{eq:Q4}
    \Hat{Q}_{\Sigma_d,x}\left(\Hat{\calS}_{\Sigma_k,R,y}\right)&=0.
\end{align}
Using \eqref{eq:Q1}, \eqref{eq:Q2}, \eqref{eq:Q3} and \eqref{eq:Q4} we get
\begin{multline}
    \Hat{Q}_{\Sigma_d,x}\left(\Hat{\calS}^\mathrm{global}_{\Sigma_k,y}\right)+\frac{1}{2}\left\{\Hat{\calS}^\mathrm{global}_{\Sigma_k,y},\Hat{\calS}^\mathrm{global}_{\Sigma_k,y}\right\}_{\Hat{\omega}^\calN_{\Sigma_k,y}}\\=\dd_y\Hat{\calS}^\mathrm{AKSZ}_{\Sigma_k,y}+\frac{1}{2}\left\{\Hat{\calS}_{\Sigma_k,R,y},\Hat{\calS}_{\Sigma_k,R,y}\right\}_{\Hat{\omega}^\calN_{\Sigma_k,y}}+\Tilde{\phi}_y^*\Hat{Q}_{\calM}\left(\Hat{\calS}^\mathrm{target}_{\Sigma_k,y}\right)+\Hat{R}_{\Sigma_d}\left(\Hat{\calS}_{\Sigma_k,R,y}\right)\\+\frac{1}{2}\left\{\Hat{\calS}^\mathrm{kin}_{\Sigma_k,y},\Hat{\calS}^\mathrm{kin}_{\Sigma_k,y}\right\}_{\Hat{\omega}^\calN_{\Sigma_k,y}}+\left\{\Hat{\calS}^\mathrm{kin}_{\Sigma_k,y},\Hat{\calS}^\mathrm{target}_{\Sigma_k,y}\right\}_{\Hat{\omega}^\calN_{\Sigma_k,y}}+\frac{1}{2}\left\{\Hat{\calS}^\mathrm{target}_{\Sigma_k,y},\Hat{\calS}^\mathrm{target}_{\Sigma_k,y}\right\}_{\Hat{\omega}^\calN_{\Sigma_k,y}}\\ =\dd_y\Hat{\calS}^\mathrm{AKSZ}_{\Sigma_k,y}+\frac{1}{2}\left\{\Hat{\calS}_{\Sigma_k,R,y},\Hat{\calS}_{\Sigma_k,R,y}\right\}_{\Hat{\omega}^\calN_{\Sigma_k,y}}+\Hat{R}_{\Sigma_d}\left(\Hat{\calS}_{\Sigma_k,R,y}\right)\\
    +\frac{1}{2}\left\{\Hat{\calS}^\mathrm{kin}_{\Sigma_k,y},\Hat{\calS}^\mathrm{kin}_{\Sigma_k,y}\right\}_{\Hat{\omega}^\calN_{\Sigma_k,y}}+\underbrace{\Hat{\dd}_{\Sigma_k}\Hat{\calS}^\mathrm{target}_{\Sigma_k,y}}_{=0}+\Tilde{\phi}^*_y\left(\Hat{Q}_\calM+\frac{1}{2}p^*\Hat{\calV}_\calE\right)\left(\Hat{\calS}^\mathrm{target}_{\Sigma_k,y}\right)\\
    =\dd_y\Hat{\calS}^\mathrm{AKSZ}_{\Sigma_k,y}+\frac{1}{2}\left\{\Hat{\calS}_{\Sigma_k,R,y},\Hat{\calS}_{\Sigma_k,R,y}\right\}_{\Hat{\omega}^\calN_{\Sigma_k,y}}+\Hat{R}_{\Sigma_d}\left(\Hat{\calS}_{\Sigma_k,R,y}\right)\\+\underbrace{(-1)^k\Tilde{\phi}^*_yp^*\mathscr{T}_{\Sigma_k}^\calE\left(Q_\calM(\Theta_\calE)+\frac{1}{2}\calV_\calE(\Theta_\calE)\right)}_{=(-1)^k\Tilde{\phi}^*_yp^*\mathscr{T}_{\Sigma_k}^\calE\left(Q_\calM(\Theta_\calE)+\frac{1}{2}\{\Theta_\calE,\Theta_\calE\}_{\omega_\calN}\right)=0 \quad(\textnormal{by definition of Hamiltonian $Q$-bundle})}\\=\dd_y\Hat{\calS}^\mathrm{AKSZ}_{\Sigma_k,y}+\frac{1}{2}\left\{\Hat{\calS}_{\Sigma_k,R,y},\Hat{\calS}_{\Sigma_k,R,y}\right\}_{\Hat{\omega}^\calN_{\Sigma_k,y}}+\Hat{R}_{\Sigma_d}\left(\Hat{\calS}_{\Sigma_k,R,y}\right)
\end{multline}
Note that $\Hat{R}_{\Sigma_d}$ is a vector field on the lifted space $\Hat{\calF}^\calM_{\Sigma_d,x}$ which implies that 
$\Hat{R}_{\Sigma_d}\left(\Hat{\calS}_{\Sigma_k,R,y}\right)=0$
because $\Hat{\calS}_{\Sigma_k,R,y}\in C^\infty\left(\Hat{\calF}^\calM_{\Sigma_d,x}\times\Hat{\calF}^\calN_{\Sigma_k,y}\right)$ is constant in the direction of $\Hat{\calF}^\calM_{\Sigma_d,x}$ and the claim follows. 
\end{proof}

\begin{cor}
An equivalent condition for the formal global AKSZ-BV theory as in Theorem \ref{thm:global_pre-observable} to be a pre-observable is given by 
\begin{equation}
\label{eq:obstruction_2}
    \Hat{\calV}^\calE_{\Sigma_k,y}\left(\Hat{\calS}_{\Sigma_k,R,y}\right)=\Hat{\calS}_{\Sigma_k,\dd_yR,y}.
\end{equation}
\end{cor}
\begin{proof}
Note that we have 
\begin{multline}
    \dd_y\Hat{\calS}^\mathrm{AKSZ}_{\Sigma_k,y}+\frac{1}{2}\left\{\Hat{\calS}_{\Sigma_k,R,y},\Hat{\calS}_{\Sigma_k,R,y}\right\}_{\Hat{\omega}^\calN_{\Sigma_k,y}}\\=\left\{\Hat{\calS}_{\Sigma_k,R,y},\Hat{\calS}_{\Sigma_kx}\right\}_{\Hat{\omega}^\calN_{\Sigma_k,y}}+\frac{1}{2}\left\{\Hat{\calS}_{\Sigma_k,R,y},\Hat{\calS}_{\Sigma_k,R,y}\right\}_{\Hat{\omega}^\calN_{\Sigma_k,y}}\\=\Hat{\calV}^\calE_{\Sigma_k,y}\left(\Hat{\calS}_{\Sigma_k,R,y}\right)+\underbrace{\frac{1}{2}\Hat{\calS}_{\Sigma_k,[R,R],x}}_{=\Hat{\calS}_{\Sigma_k,\dd_yR,y}},
\end{multline}
where we have used that $\Hat{\calV}^\calE_{\Sigma_k,y}$ is the Hamiltonian vector field of $\Hat{\calS}^\mathrm{AKSZ}_{\Sigma_k,y}$ and the fact that \cite{BCM}
\begin{equation}
    \left\{\Hat{\calS}_{\Sigma_k,R,y},\Hat{\calS}_{\Sigma_k,R,y}\right\}_{\Hat{\omega}_{\Sigma_k,y}^\calN}=\Hat{\calS}_{\Sigma_k,[R,R],y}.
\end{equation}
the last equality (under the braces) follows from the fact that $D$ is a flat connection on $\Hat{\Sym}(T^*\calN)$ which can be translated into 
\begin{equation}
    \dd_yR+\frac{1}{2}[R,R]=0.
\end{equation}
Moreover, it is easy to see that $\Hat{\calS}_{\Sigma_k,\ell R,y}=\ell\Hat{\calS}_{\Sigma_k,R,y}$ for any $\ell\in \R$.
\end{proof}

\subsection{Formal global auxiliary construction in coordinates}
We want to describe the auxiliary theory as well as its formal global extension in terms of coordinates. The description follows similarly from the description of the ambient theory as in Section \ref{sec:AKSZ_sigma_models} and its formal global extension as in Section \ref{sec:formal_global_AKSZ_models}.
Let $(v^j)$ be even local coordinates on $\Sigma_k$ and consider the corresponding odd local coordinates $\xi^j=\dd_{\Sigma_k} v^j$ for $1\leq j\leq k$. Then we can construct superfield coordinates 
\begin{equation}
    \calB^\nu(v,\xi)=\sum_{\ell=1}^k\,\,\underbrace{\sum_{1\leq j_1<\dotsm <j_\ell\leq k}\calB^\nu_{j_1\ldots j_\ell}(v)\xi^{j_1}\land\dotsm \land \xi^{j_\ell}}_{=\calB^\nu_{(\ell)}(v,\xi)}\in\bigoplus_{\ell=0}^k C^\infty(\Sigma_k)\otimes \bigwedge^\ell T^*\Sigma_k.
\end{equation}
associated to local homogeneous coordinates $(y^\nu)$ of $\calN$. Note that locally we have 
\begin{align}
    \alpha_\calN&=\alpha_\nu^\calN(y)\dd_\calN y^\nu\in\Omega^1(\calN),\\
    \omega_\calN&=\frac{1}{2}\omega^\calN_{\nu_1\nu_2}(y)\dd_\calN y^{\nu_1}\land \dd_\calN y^{\nu_2}\in\Omega^2(\calN).
\end{align}
Hence we get 
\begin{align}
    \alpha_{\Sigma_k}^\calN&=\int_{\Sigma_k}\alpha^\calN_\nu(\calB)\delta\calB^\nu\in\Omega^1\left(\calF_{\Sigma_k}^\calN\right),\\
    \omega_{\Sigma_k}^\calN&=(-1)^k\frac{1}{2}\int_{\Sigma_k}\omega^\calN_{\nu_1\nu_2}(\calB)\delta\calB^{\nu_1}\land \delta\calB^{\nu_2}\in\Omega^2\left(\calF_{\Sigma_k}^\calN\right),
\end{align}
and thus we get an action for the auxiliary fields as 
\begin{equation}
    \calS_{\Sigma_k}^\calN(\calA,\calB;i)=\int_{\Sigma_k}\alpha^\calN_\nu(\calB)\dd_{\Sigma_k}\calB^\nu+\int_{\Sigma_k}\Theta_\calE(i^*\calA,\calB)\in C^\infty\left(\calF^\calM_{\Sigma_d}\times\calF_{\Sigma_k}^\calN\right),
\end{equation}
These expressions can be lifted to the formal global construction. Indeed, consider a formal exponential map $\phi\colon T\calN\to \calN$. Let $\Hat{\calA}=\phi_x^{-1}(\calA)$ be the lift of $\calA$ to $\Hat{\calF}^\calM_{\Sigma_d,x}$ and $\Hat{\calB}=\phi_y^{-1}(\calB)$ be the lift of $\calB$ to $\Hat{\calF}^\calN_{\Sigma_k,y}$ for $x\in \calM$ and $y\in\calN$. Then we get 
\begin{align}
    \Hat{\alpha}_{\Sigma_k,y}^\calN&=\int_{\Sigma_k}\Hat{\alpha}_\nu\left(\Hat{\calB}\right)\delta\Hat{\calB}^\nu\in \Omega^1\left(\Hat{\calF}^\calN_{\Sigma_k,y}\right),\\
    \Hat{\omega}_{\Sigma_k,y}^\calN&=(-1)^k\frac{1}{2}\int_{\Sigma_k}\Hat{\omega}_{\nu_1\nu_2}\left(\Hat{\calB}\right)\delta\Hat{\calB}^{\nu_1}\land \delta\Hat{\calB}^{\nu_2}\in \Omega^2\left(\Hat{\calF}^\calN_{\Sigma_k,y}\right),
\end{align}
where $\Hat{\alpha}^\calN_{\nu}$ and $\Hat{\omega}^\calN_{\nu_1\nu_2}$ are the coefficients of $\Hat{\alpha}_\calN\in\Omega^1(T\calN)$ and $\Hat{\omega}_\calN\in \Omega^2(T\calN)$ respectively. If we set $\Hat{\Theta}_{\calE,y}:=\mathsf{T}\Tilde{\phi}^*_y\Theta_\calE$, the auxiliary formal global AKSZ action is then given by 
\begin{multline}
    \Hat{\calS}_{\Sigma_k,y}^\mathrm{global}\left(\Hat{\calA},\Hat{\calB};i\right)=\underbrace{\int_{\Sigma_k}\Hat{\alpha}_\nu^\calN\left(\Hat{\calB}\right) \dd_{\Sigma_k}\Hat{\calB}^\nu+\int_{\Sigma_k}\Hat{\Theta}_{\calE,y}\left(i^*\Hat{\calA},\Hat{\calB}\right)}_{=\Hat{\calS}^\mathrm{AKSZ}_{\Sigma_k,y}}\\+\underbrace{\int_{\Sigma_k}(\Hat{R}_{\Sigma_k})_{\nu}\left(y,\Hat{\calB}\right)\dd_\calN y^{\nu}}_{=\Hat{\calS}_{\Sigma_k,R,y}}.
\end{multline}

\subsection{Formal global split auxiliary construction in coordinates}
If we consider a split AKSZ model with target $\calM=T^*[d-1]M$, for some graded manifold $M$, for the ambient theory associated to $\Sigma_d$, we can consider a split construction for the auxiliary theory associated to the embedding $i\colon \Sigma_k\hookrightarrow \Sigma_d$. We set $\calN=T^*[k-1]N$ for some graded manifold $N$. The description is analogously given by the one of the ambient theory as in Section \ref{sec:formal_global_split_AKSZ_construction}.
Hence we have 
\begin{equation}
    \calF_{\Sigma_k}^\calN=\Map(T[1]\Sigma_k,T^*[k-1]N),
\end{equation}
and choosing a formal exponential map $\phi\colon TN\to N$ together with $y\in N$ we get
\begin{align}
\begin{split}
    \Hat{\calF}^\calN_{\Sigma_k,y}&=\Map(T[1]\Sigma_k,T^*[k-1]T_yN)\\
    &\cong \Omega^\bullet(\Sigma_k)\otimes T_yN\oplus\Omega^\bullet(\Sigma_k)\otimes T^*_yN[k-1].
\end{split}
\end{align}
Then we can write $\Hat{\calA}=\left(\Hat{\boldsymbol{A}},\Hat{\boldsymbol{B}}\right)\in \Hat{\calF}^\calM_{\Sigma_d,x}$ and $\Hat{\calB}=\left(\Hat{\boldsymbol{\alpha}},\Hat{\boldsymbol{\beta}}\right)\in \Hat{\calF}^\calN_{\Sigma_k,y}$, thus we have an auxiliary formal global split AKSZ action given by 
\begin{multline}
    \Hat{\calS}^\mathrm{global}_{\Sigma_k,y}\left(\Hat{\boldsymbol{A}},\Hat{\boldsymbol{B}},\Hat{\boldsymbol{\alpha}},\Hat{\boldsymbol{\beta}};i\right)=\int_{\Sigma_k}\Hat{\boldsymbol{\beta}}_\ell\land\dd_{\Sigma_k}\Hat{\boldsymbol{\alpha}}^\ell+\int_{\Sigma_k}\Hat{\Theta}_{\calE,y}\left(i^*\Hat{\boldsymbol{A}},i^*\Hat{\boldsymbol{B}},\Hat{\boldsymbol{\alpha}},\Hat{\boldsymbol{\beta}}\right)\\+\int_{\Sigma_k}R_\ell^j\left(y,\Hat{\boldsymbol{\alpha}}\right)\Hat{\boldsymbol{\beta}}_j\land\dd_{N} y^\ell,
\end{multline}
where $R\in \Omega^1\left(N,\mathrm{Der}\left(\Hat{\Sym}(T^*N)\right)\right)$.

\section{From pre-observables to observables}

\subsection{AKSZ-observables}
We want to construct the observables for the AKSZ theories out of pre-obsrvables by integrating out means of auxiliary fields similarly as in Proposition \ref{prop:from_pre-observable_to_observable}. For a submanifold $i\colon \Sigma_k\hookrightarrow \Sigma_d$ we set 
\begin{equation}
    \calO_{\Sigma_k}(\calA,\calB;i)=\int_{\calL\subset \calF_{\Sigma_k}}\mathscr{D}[\calB]\ee^{\frac{\I}{\hbar}\calS^\calN_{\Sigma_k}(\calA,\calB;i)}\in C^\infty\left(\calF^\calM_{\Sigma_d}\right).
\end{equation}
There are several things to note. First, $\calO_{\Sigma_k}$ depends only on the fields in $\calF^\calM_{\Sigma_d}$ via the pullback of $i\colon \Sigma_k\hookrightarrow \Sigma_d$, hence $Q_{\Sigma_d}(\calO_{\Sigma_k})=0$ which is consistent with the definition of an observable. Moreover, the $Q_{\Sigma_k}$-cohomology class of $\calO_{\Sigma_k}$ does not depend on deformations of the Lagrangian submanifold $\calL\subset\calF_{\Sigma_k}$ and is invariant under isotopies of $\Sigma_k$. We get the following Proposition. 
\begin{prop}
\label{prop:invariance_1}
Let $\mathrm{Diff}_0(\Sigma_k)\subset \mathrm{Diff}(\Sigma_k)$ be diffeomorphisms on $\Sigma_k$ which are connected to the identity. Then for $\varphi_k\in \mathrm{Diff}(\Sigma_k)$ we have
\begin{equation}
    \calO_{\Sigma_k}(\calA,\calB;i\circ \varphi_k)=\calO_{\Sigma_k}(\calA,\calB;i)+\textnormal{$Q_{\Sigma_k}$-exact}.
\end{equation}
\end{prop}
\begin{proof}
Indeed, we have 
\begin{multline}
    \calO_{\Sigma_k}(\calA,\calB;i\circ\varphi_k)=\int_{\calL\subset \calF_{\Sigma_k}}\mathscr{D}[\calB]\ee^{\frac{\I}{\hbar}\calS^\calN_{\Sigma_k}(\calA,\calB;i\circ\varphi_k)}=\int_{\calL}\mathscr{D}[\calB]\ee^{\frac{\I}{\hbar}\calS^\calN_{\Sigma_k}\left(\calA,(\varphi_k^{-1})^*\calB;i\right)}\\ =\int_{(\varphi_k^{-1})^*\calL}(\varphi_k^{-1})^*\mathscr{D}[\calB]\ee^{\frac{\I}{\hbar}\calS^\calN_{\Sigma_k}(\calA,\calB;i)}=\int_{(\varphi_k^{-1})^*\calL}\mathscr{D}[\calB]\ee^{\frac{\I}{\hbar}\calS^\calN_{\Sigma_k}(\calA,\calB;i)}\\=\int_{\calL}\mathscr{D}[\calB]\ee^{\frac{\I}{\hbar}\calS^\calN_{\Sigma_k}(\calA,\calB;i)}+\textnormal{$Q_{\Sigma_k}$-exact}=\calO_{\Sigma_k}(\calA,\calB;i)+\textnormal{$Q_{\Sigma_k}$-exact},
\end{multline}
where we think of $\int_\calL\mathscr{D}[\calB]$ to be in fact given by $\int_\calL\sqrt{\mu}\vert_\calL$, with $\mu$ being the functional integral measure on $\calF_{\Sigma_k}^\calN$. Moreover, we have used the isotopy property of $\varphi_k$ to make sure that $\calL$ and $(\varphi_k^{-1})^*\calL$ are indeed homotopic.
\end{proof}

There is a similar invariance result for diffeomorphisms of the ambient manifold $\Sigma_d$ which is the content of the following Proposition.   

\begin{prop}
\label{prop:invariance_2}
For a diffeomorphism $\varphi_d\in \mathrm{Diff}(\Sigma_d)$ we get 
\begin{equation}
    \calO_{\Sigma_k}(\calA,\calB;\varphi_d\circ i)=\calO_{\Sigma_k}(\varphi_d^*\calA,\calB;i).
\end{equation}
\end{prop}
This is indeed true since $\calO_{\Sigma_k}$ only depends of the ambient field $\calA$ via the pullback by $i$.

Another important property is that the correlator of an observable should be invariant under ambient isotopies. 

\begin{prop}
For $\varphi_d\in\mathrm{Diff}_0(\Sigma_d)$ we have 
\begin{equation}
    \langle \calO_{\Sigma_k}(\calA,\calB;\varphi_d\circ i)\rangle=\langle \calO_{\Sigma_k}(\calA,\calB;i)\rangle.
\end{equation}
\end{prop}
\begin{proof}
Indeed, we have 
\begin{multline}
    \langle\calO_{\Sigma_k}(\calA,\calB;\varphi_d\circ i)\rangle=\int_{\calL\subset\calF_{\Sigma_k}}\mathscr{D}[\calA]\calO_{\Sigma_k}(\calA,\calB;\varphi_d\circ i)\ee^{\frac{\I}{\hbar}\calS_{\Sigma_d}(\calA)}\\=\int_{\calL}\mathscr{D}[\calA]\calO_{\Sigma_k}(\varphi_d^*\calA,\calB;i)\ee^{\frac{\I}{\hbar}\calS_{\Sigma_d}(\calA)}=\int_{\varphi_d^*\calL}(\varphi_d^*)_*\mathscr{D}[\calA]\calO_{\Sigma_k}(\calA,\calB;i)\ee^{\frac{\I}{\hbar}\calS_{\Sigma_d}((\varphi^{-1})^*\calA)}\\=\int_{\varphi_d^*\calL}\mathscr{D}[\calA]\calO_{\Sigma_k}(\calA,\calB;i)\ee^{\frac{\I}{\hbar}\calS_{\Sigma_d}(\calA)}=\int_{\calL}\mathscr{D}[\calA]\calO_{\Sigma_k}(\calA,\calB;i)\ee^{\frac{\I}{\hbar}\calS_{\Sigma_d}(\calA)}\\=\langle\calO_{\Sigma_k}(\calA,\calB;i)\rangle,
\end{multline}
where we have used Proposition \ref{prop:invariance_2} and that the AKSZ action $\calS_{\Sigma_d}$ (see Remark \ref{rem:invariance_AKSZ_diffeo}) and the functional integral measure $\mathscr{D}[\calA]$ are invariant under diffeomorphisms for our theory is topological. 
\end{proof}

\subsection{Formal global AKSZ-observables}
The construction above can be extended to a formal global one if we start with a formal global pre-observable. Then we have 
\begin{equation}
    \Hat{\calO}_{\Sigma_k,y}\left(\Hat{\calA},\Hat{\calB};i\right)=\int_{\Hat{\calL}\subset \Hat{\calF}^\calN_{\Sigma_k,y}}\mathscr{D}\left[\Hat{\calB}\right]\ee^{\frac{\I}{\hbar}\Hat{\calS}^\mathrm{global}_{\Sigma_k,y}\left(\Hat{\calA},\Hat{\calB};i\right)}\in C^\infty\left(\Hat{\calF}^\calM_{\Sigma_d}\right).
\end{equation}
If we start with a split AKSZ theory we get 
\begin{equation}
    \Hat{\calO}_{\Sigma_k,y}\left(\Hat{\boldsymbol{A}},\Hat{\boldsymbol{B}};i\right)=\int_{\Hat{\calL}\subset \Hat{\calF}^\calN_{\Sigma_k,y}}\mathscr{D}\left[\Hat{\boldsymbol{\alpha}}\right]\mathscr{D}\left[\Hat{\boldsymbol{\beta}}\right]\ee^{\frac{\I}{\hbar}\Hat{\calS}^\mathrm{global}_{\Sigma_k,y}\left(\Hat{\boldsymbol{A}},\Hat{\boldsymbol{B}},\Hat{\boldsymbol{\alpha}},\Hat{\boldsymbol{\beta}};i\right)}\in C^\infty\left(\Hat{\calF}^\calM_{\Sigma_d}\right).
\end{equation}
We have the following proposition (quantum version of \eqref{eq:dCME}).
\begin{prop}[dQME]
\label{prop:dQME}
The \emph{differential Quantum Master Equation} (dQME) for the formal global split AKSZ-observable holds:
\begin{equation}
\label{eq:dQME}
\boxed{
    \dd_y\Hat{\calO}_{\Sigma_k,y}-(-1)^d\I\hbar\Delta\Hat{\calO}_{\Sigma_k,y}=0.}
\end{equation}
\end{prop}
\begin{proof}
Note that we have 
\begin{equation}
    \dd_y\Hat{\calO}_{\Sigma_k,y}=-\frac{\I}{\hbar}\int_{\Hat{\calL}\subset \Hat{\calF}^\calN_{\Sigma_k,y}}\mathscr{D}\left[\Hat{\boldsymbol{\alpha}}\right]\mathscr{D}\left[\Hat{\boldsymbol{\beta}}\right]\ee^{\frac{\I}{\hbar}\Hat{\calS}^\mathrm{global}_{\Sigma_k,y}\left(\Hat{\boldsymbol{A}},\Hat{\boldsymbol{B}},\Hat{\boldsymbol{\alpha}},\Hat{\boldsymbol{\beta}};i\right)}\left\{\Hat{\calS}_{\Sigma_k,R,y},\Hat{\calS}^\mathrm{AKSZ}_{\Sigma_k,y}\right\}_{\omega^\calN_{\Sigma_k,y}},
\end{equation}
which we can write as 
\begin{multline}
    -\frac{\I}{\hbar}\int_{\Hat{\calL}\subset \Hat{\calF}^\calN_{\Sigma_k,y}}\mathscr{D}\left[\Hat{\boldsymbol{\alpha}}\right]\mathscr{D}\left[\Hat{\boldsymbol{\beta}}\right]\ee^{\frac{\I}{\hbar}\Hat{\calS}^\mathrm{global}_{\Sigma_k,y}\left(\Hat{\boldsymbol{A}},\Hat{\boldsymbol{B}},\Hat{\boldsymbol{\alpha}},\Hat{\boldsymbol{\beta}};i\right)}\left\{\Hat{\calS}_{\Sigma_k,R,y},\Hat{\calS}^\mathrm{AKSZ}_{\Sigma_k,y}\right\}_{\omega^\calN_{\Sigma_k,y}}\\=-\Delta\int_{\Hat{\calL}\subset \Hat{\calF}^\calN_{\Sigma_k,y}}\mathscr{D}\left[\Hat{\boldsymbol{\alpha}}\right]\mathscr{D}\left[\Hat{\boldsymbol{\beta}}\right]\ee^{\frac{\I}{\hbar}\Hat{\calS}^\mathrm{global}_{\Sigma_k,y}\left(\Hat{\boldsymbol{A}},\Hat{\boldsymbol{B}},\Hat{\boldsymbol{\alpha}},\Hat{\boldsymbol{\beta}};i\right)}\Hat{\calS}_{\Sigma_k,R,y}
\end{multline}
if we assume that $\Delta\Hat{\calS}_{\Sigma_k,R,y}=0$, which is true e.g. if the Euler characteristic of $\Sigma_k$ is zero or if $\Div_{\mathsf{T}\phi^*\mu}R=0$, where $\mu$ is some volume form on $\calN$. Note that $\dd_y\mathsf{T}\phi^*\mu=-L_R\mathsf{T}\phi^*\mu$ which means that $\Div_{\mathsf{T}\phi^*\mu}R=0$ if and only if $\dd_y\mathsf{T}\phi^*\mu=0$. For any volume element $\mu$ it is always possible to find a formal exponential map $\phi$ such that the latter condition is satisfied. Note that this is then also translated into the \emph{differential Quantum Master Equation} 
\begin{equation}
    \dd_y\Hat{\calS}^\mathrm{global}_{\Sigma_k,y}+\frac{1}{2}\left\{\Hat{\calS}^\mathrm{global}_{\Sigma_k,y},\Hat{\calS}^\mathrm{global}_{\Sigma_k,y}\right\}_{\Hat{\omega}_{\Sigma_k,y}}-\I\hbar\Delta\Hat{\calS}^\mathrm{global}_{\Sigma_k,y}=0,
\end{equation}
and by the assumption $\Delta\Hat{\calS}^\mathrm{global}_{\Sigma_k,y}=0$, we obtain the differential CME as in \eqref{eq:dCME}. Hence the claim follows.
\end{proof}

\begin{rem}
One can check that $\Hat{Q}_{\Sigma_d,x}\left(\Hat{\calO}_{\Sigma_k,y}\right)=0$ and that Proposition \ref{prop:invariance_1} and \ref{prop:invariance_2} also hold for the formal global extension if we indeed start with a formal global pre-observable, i.e. that the assumption of Theorem \ref{thm:global_pre-observable} is satisfied.
\end{rem}

\begin{rem}
The dQME as in \eqref{eq:dQME} can be thought of as a descent equation for different form degrees. In fact we have 
\begin{equation}
    \Hat{\delta}_{\mathrm{BV}}\Hat{\calO}_{\Sigma_k,y}=(-1)^d\dd_y\Hat{\calO}_{\Sigma_k,y}
\end{equation}
since $\Hat{\calO}_{\Sigma_k,y}$ is a formal global observable. We have set $\Hat{\delta}_{\mathrm{BV}}=\Hat{Q}_{\Sigma_d,x}-\I\hbar\Delta$.
\end{rem}

\begin{rem}
Note that if $\calN$ is a point, we have $\calV_\calE=0$, $\omega^\calN_{\Sigma_k}=0$ and $\Theta_\calE\in C^\infty(\calM)$. The associated pre-observable is then given by 
\begin{equation}
    \calF^\calN_{\Sigma_k}=\mathrm{pt},\quad \calV^\calE_{\Sigma_k}=0,\quad \omega_{\Sigma_k}^\calN=0,\quad \calS^\calN_{\Sigma_k}(\calA)=\int_{\Sigma_k}\Theta_\calE(i^*\calA).
\end{equation}
Hence, since there are no auxiliary fields $\calB$, the constructed observable is given by 
\begin{equation}
    \calO_{\Sigma_k}(\calA;i)=\ee^{\frac{\I}{\hbar}\int_{\Sigma_k}\Theta_\calE(i^*\calA)}.
\end{equation}
This can be easily lifted to a formal global pre-observable by 
\begin{equation}
    \Hat{\calF}^\calN_{\Sigma_k,y}=\mathrm{pt},\quad \Hat{\calV}^\calE_{\Sigma_k,y}=0,\quad \Hat{\omega}_{\Sigma_k,y}^\calN=0,\quad \Hat{\calS}^\mathrm{global}_{\Sigma_k,y}\left(\Hat{\calA}\right)=\int_{\Sigma_k}\Hat{\Theta}_{\calE,y}\left(i^*\Hat{\calA}\right),
\end{equation}
Thus, we get a formal global observable by 
\begin{equation}
    \Hat{\calO}_{\Sigma_k,y}\left(\Hat{\calA};i\right)=\ee^{\frac{\I}{\hbar}\int_{\Sigma_k}\Hat{\Theta}_{\calE,y}\left(i^*\Hat{\calA}\right)}.
\end{equation}
\end{rem}

\subsection{Loop observables}
\label{subsec:loops_observables}
Let us consider the case where $S^1$ is embedded into $\Sigma_d$, i.e. $i\colon \Sigma_1:=S^1\hookrightarrow \Sigma_d$ and assume that $\calN$ is given by an ordinary symplectic manifold with symplectic structure $\omega_\calN=\dd_\calN\alpha_\calN$, which means that $\calN$ is concentrated in degree zero. 
Let $\sigma$ denote the coordinate on $\Sigma_1$. Then we can write the auxiliary field as 
\begin{equation}
    \calB^\nu(\sigma,\dd_{\Sigma_1}\sigma)=\calB^\nu_{(0)}(\sigma)+\calB^\nu_{(1)}\dd_{\Sigma_1}\sigma,
\end{equation}
and hence we get a pre-observable by
\begin{align}
    \label{eq:space_of_fields_loop}
    \calF^\calN_{\Sigma_1}&=\Map_{\mathrm{GrMnf}}(T[1]\Sigma_1,\calN),\\
    \begin{split}
    \label{eq:omega_loop}
    \omega^\calN_{\Sigma_1}&=-\oint_{\Sigma_1}\omega^\calN_{\nu_1\nu_2}\left(\calB_{(0)}\right)\delta\calB_{(0)}^{\nu_1}\land \delta\calB_{(1)}^{\nu_2}\\
    &+\oint_{\Sigma_1}\frac{1}{2}\calB^{\nu_3}_{(1)}\de_{\nu_3}\omega_{\nu_1\nu_2}^\calN\left(\calB_{(0)}\right)\delta\calB_{(0)}^{\nu_1}\land \delta\calB_{(0)}^{\nu_2},\end{split}\\
    \calS^\calN_{\Sigma_1}&=\oint_{\Sigma_1}\alpha^\calN_\nu\left(\calB_{(0)}\right)\dd_{\Sigma_1}\calB^\nu_{(0)}+\oint_{\Sigma_1}\Theta_{\calE}(i^*\calA,\calB),
\end{align}
where $\omega^\calN_{\nu_1\nu_2}$ are the coefficients of $\omega_\calN$ and $\alpha^\calN_\nu$ are the coefficients of $\alpha_\calN$. Note that in this setting we have
\begin{equation}
    \calF^\calN_{\Sigma_1}=\left\{\left(\calB_{(0)},\calB_{(1)}\right)\,\Big|\, \calB_{(0)}\colon \Sigma_1\to \calN,\,\, \calB_{(1)}\in \Gamma\left(\Sigma_1,T^*\Sigma_1\otimes \calB_{(0)}^*T^*\calN\right)[-1]\right\}.
\end{equation}
Hence we can construct the observable as 
\begin{equation}
    \calO_{\Sigma_1}(\calA;i)=\int_{\calL}\mathscr{D}\left[\calB_{(0)}\right]\ee^{\frac{\I}{\hbar}\oint_{\Sigma_1}\alpha_\nu^\calN(\calB_{(0)})\dd_{\Sigma_1}\calB_{(0)}+\frac{\I}{\hbar}\oint_{\Sigma_1}\Theta_{\calE}\left(i^*\calA,\calB_{(0)}\right)},
\end{equation}
where we have chosen the natural Lagrangian submanifold 
\begin{equation}
\calL=\Map_{\mathrm{Mnf}}(\Sigma_1,\calN)\subset\calF^\calN_{\Sigma_1}, 
\end{equation}
which is obtained by setting all odd variables $\calB_{(1)}$ to zero.

\begin{rem}[Bohr--Sommerfeld]
\label{rem:Bohr-Sommerfeld}
If $\calN$ is a differential graded symplectic manifold of degree different from zero, we know that the symplectic form $\omega_\calN$ is always exact since we can write it as \[\omega_\calN=\dd_\calN(\iota_E\omega_\calN),\]
(see \cite{Roytenberg2005}) where $E$ is the \emph{Euler vector field}. For the degree zero case, the symplectic form does not automatically have a primitive 1-form and hence one can not immediately define $\calS^\mathrm{kin}_{\Sigma_1}$. However, one can also assume that $\omega_\calN$ satisfies the \emph{Bohr--Sommerfeld condition}, which says that \[\frac{\omega_\calN}{2\pi}\in H^2(\calN,\mathbb{Z}).\]
Then the primitive 1-form can be understood as a Hermitian line bundle over $\calN$ endowed with a $U(1)$-connection $\nabla_\calN$ such that its curvature is given by $(\nabla_\calN)^2=\omega_\calN$. Thus we can define
\[\ee^{\frac{\I}{\hbar}\calS_{\Sigma_1}^\mathrm{kin}(\calB)}\] to be given by the holonomy of $\left(\calB_{(0)}\right)^*\nabla_\calN$ around $\Sigma_1$. Using Stokes' theorem we get
\begin{equation}
    \calS_{\Sigma_1}^{\mathrm{kin}}(\calB)=\int_{\mathbb{D}}\left(\calB_{(0)}^\mathrm{ext}\right)^*\omega_{\calN},
\end{equation}
where $\mathbb{D}$ is a disk with $\de\mathbb{D}=\Sigma_1$ and $\calB_{(0)}^\mathrm{ext}$ is any extension of $\calB_{(0)}$ to $\mathbb{D}$.
\end{rem}

\begin{rem}
This construction can be obviously extended to the formal global case. The case of a dimension 1 submanifold gives the same auxiliary theory as for the case when our theory is split. 
\end{rem}

\subsection{Formal global loop observables}

The following proposition is an extension of Proposition 5 in \cite{Mn3} to the formal global case.
\begin{prop}
\label{prop:formal_gloabl_wilson_loop}
Let $(\calN,\omega_\calN)$ be a symplectic manifold and assume that it can be geometrically quantized to a complex vector space $\calH$, the state space, and that the Hamiltonian $\Theta_\calE\in C^\infty(\calE)$ can be quantized to an operator valued function $\boldsymbol{\Theta}_\calE\in C^\infty(\calM)\otimes \mathrm{End}(\calH)$. Moreover, for a formal exponential map $\phi\colon T\calM\to \calM$, let $\Hat{\boldsymbol{\Theta}}_{\calE,x}:=\mathsf{T}\Tilde{\phi}^*_x\boldsymbol{\Theta}_\calE$ and assume that 
\begin{equation}
\label{eq:quantum_1}
    \Hat{Q}_{\calM}\left(\Hat{\boldsymbol{\Theta}}_{\calE,x}\right)+\Hat{R}_{\Sigma_d}\left(\Hat{\boldsymbol{\Theta}}_{\calE,x}\right)+\I\hbar\left(\Hat{\boldsymbol{\Theta}}_{\calE,x}\right)^2=0
\end{equation}
for $x\in\calM$. Then for $\Sigma_1:=S^1$ we get that
\begin{equation}
\label{eq:assumption_1}
    \Hat{\calO}_{\Sigma_1,x}=\tr_\calH \calP\exp\left(\frac{\I}{\hbar}\oint_{\Sigma_1}\Hat{\boldsymbol{\Theta}}_{\calE,x}\left(i^*\Hat{\calA}\right)\right)
\end{equation}
is a formal global observable, where we have denoted by $\tr_\calH$ the trace map on $\calH$ and $\calP\exp$ denotes the path-ordered exponential. 
\end{prop}
\begin{rem}
    Note that \eqref{eq:quantum_1} is the formal global quantum version of \eqref{eq:classical_1}.
\end{rem}

\begin{proof}[Proof of Proposition \ref{prop:formal_gloabl_wilson_loop}]
Let $\gamma\colon\Sigma_1:=[0,1]\to \Sigma_{d}$ be a path in $\Sigma_d$ which is parametrized by $t\in[0,1]$. Denote by $$\Hat{\boldsymbol{\psi}}:=\Hat{\boldsymbol{\Theta}}_{\calE,x}\left(\gamma^*\Hat{\calA}\right)\in\Omega^\bullet([0,1])\otimes C^\infty\left(\Hat{\calF}^\calM_{\Sigma_d,x}\right)\otimes \End(\calH).$$ 
Moreover denote by $\Hat{\boldsymbol{\psi}}_{(0)}(t)$ and $\Hat{\boldsymbol{\psi}}_{(1)}(t,\dd t)$ the 0- and 1-form part of $\Hat{\boldsymbol{\psi}}$. Then, for the 1-form part, we get 
\begin{multline}
    \Hat{W}_{\Sigma_1,x}^\gamma=\calP\exp\left(\frac{\I}{\hbar}\int_0^1\Hat{\boldsymbol{\psi}}_{(1)}\right)\\
    =\lim_{N\to\infty}\overleftarrow{\prod_{0\leq r\leq N}}\left(\id_\calH+\frac{\I}{\hbar}\iota_{\frac{1}{N}\frac{\de}{\de t}}\Hat{\boldsymbol{\psi}}_{(1)}\left(\frac{r}{N},\dd t\right)\right)\in C^\infty\left(\Hat{\calF}_{\Sigma_d,x}\right)\otimes\End(\calH)
\end{multline}
Then we get 
\begin{multline}
    \Hat{Q}_{\Sigma_d,x}\left(\Hat{W}^\gamma_{\Sigma_1,x}\right)=-\I\hbar\int_0^1\calP\exp\left(\frac{\I}{\hbar}\int_t^1\Hat{\boldsymbol{\psi}}_{(1)}\right)\Hat{Q}_{\Sigma_d,x}\left(\Hat{\boldsymbol{\psi}}(t,\dd t)\right)\calP\exp\left(\frac{\I}{\hbar}\int_0^t\boldsymbol{\psi}_{(1)}\right)\\
    =-\I\hbar\int_0^1\calP\exp\left(\frac{\I}{\hbar}\int_t^1\Hat{\boldsymbol{\psi}}_{(1)}\right)\left(\dd t\frac{\de}{\de t}\Hat{\boldsymbol{\psi}}_{(0)}(t)-\I\hbar\left[\Hat{\boldsymbol{\psi}}_{(0)}(t),\Hat{\boldsymbol{\psi}}_{(1)}(t,\dd t)\right]\right)\calP\exp\left(\frac{\I}{\hbar}\int_0^t\Hat{\boldsymbol{\psi}}_{(1)}\right)\\
    =-\I\hbar\lim_{N\to\infty}\sum_{\ell=0}^{N-1}\overleftarrow{\prod_{\ell<r<N}}\left(\id_\calH+\frac{\I}{\hbar}\iota_{\frac{1}{N}\frac{\de}{\de t}}\Hat{\boldsymbol{\psi}}_{(1)}\left(\frac{r}{N},\dd t\right)\right)\times\\
    \times\left(\Hat{\boldsymbol{\psi}}_{(0)}\left(\frac{\ell+1}{N}\right)\left(\id_\calH+\frac{\I}{\hbar}\iota_{\frac{1}{N}\frac{\de}{\de t}}\Hat{\boldsymbol{\psi}}_{(1)}\left(\frac{\ell}{N},\dd t\right)\right)-\left(\id_\calH+\frac{\I}{\hbar}\iota_{\frac{1}{N}\frac{\de}{\de t}}\Hat{\boldsymbol{\psi}}_{(1)}\left(\frac{\ell}{N},\dd t\right)\right)\Hat{\boldsymbol{\psi}}_{(0)}\left(\frac{\ell}{N}\right)\right)\times\\
    \times \overleftarrow{\prod_{0\leq r<\ell}}\left(\id_\calH+\frac{\I}{\hbar}\iota_{\frac{1}{N}\frac{\de}{\de t}}\Hat{\boldsymbol{\psi}}_{(1)}\left(\frac{r}{N},\dd t\right)\right)=-\I\hbar\left(\Hat{\boldsymbol{\psi}}_{(0)}(1)\Hat{W}_{\Sigma_1,x}^\gamma-\Hat{W}_{\Sigma_1,x}^\gamma\Hat{\boldsymbol{\psi}}_{(0)}(0)\right).
\end{multline}
We have used \eqref{eq:assumption_1}, which gives us 
\begin{equation}
    \Hat{Q}_{\Sigma_d,x}\left(\Hat{\boldsymbol{\psi}}\right)=\dd_{\Sigma_1}\Hat{\boldsymbol{\psi}}-\I\hbar\left[\Hat{\boldsymbol{\psi}},\Hat{\boldsymbol{\psi}}\right],
\end{equation}
where $[\enspace,\enspace]$ denotes the commutator of operators.
Now if $\Sigma_1:=S^1$ we have $\gamma(0)=\gamma(1)$, and thus we get 
\begin{equation}
    \Hat{Q}_{\Sigma_d,x}\left(\Hat{\calO}_{\Sigma_1,x}\right)=\tr_\calH\Hat{Q}_{\Sigma_d,x}\left(\Hat{W}_{\Sigma_1,x}^\gamma\right)=-\I\hbar\tr_\calH\left[\Hat{\boldsymbol{\Theta}}_{\calE,x}\left(\Hat{\calA}_{(0)}(\gamma(0))\right),\Hat{W}_{\Sigma_1,x}^\gamma\right]=0,
\end{equation}
where $\Hat{\calA}_{(0)}$ denotes the degree zero component of $\Hat{\calA}$.
\end{proof}

\begin{rem}
The construction in Proposition \ref{prop:formal_gloabl_wilson_loop} does not require $\omega_\calN$ to be exact. It is in fact enough to require that $\omega_\calN$ satisfies the Bohr--Sommerfeld condition as discussed in Remark \ref{rem:Bohr-Sommerfeld}. This is necessary for the assumption that $\calN$ can be geometrically quantized. 
\end{rem}

\subsection{Formal global loop observables for the Poisson sigma model}
The Poisson sigma model is an example of a 2-dimensional AKSZ theory which is split as in Definition \ref{defn:split_AKSZ_sigma_model}. Let $M$ be a Poisson manifold with Poisson bivector $\pi\in \Gamma\left(\bigwedge^2TM\right)$. Moreover, consider a 2-dimensional source $\Sigma_2$. 
Let $(x,p)$ be base and fiber coordinates on $T^*[1]M$.
Then we can define a differential graded symplectic manifold as the target of the AKSZ theory by the data
\begin{align}
    \calM&=T^*[1]M,\\
    Q_\calM&=\left\langle\pi(x),p\frac{\de}{\de x}\right\rangle+\frac{1}{2}\left\langle \frac{\de}{\de x}\pi(x),(p\land p)\otimes\frac{\de}{\de p}\right\rangle,\\
    \omega_\calM&=\langle \delta p,\delta x\rangle,\\
    \alpha_\calM&=\langle p,\delta x\rangle,\\
    \Theta_\calM&=\frac{1}{2}\langle \pi(x),p\land p\rangle.
\end{align}
The corresponding 2-dimensional AKSZ-BV theory is given by the data
\begin{align}
    \begin{split}
    \calF_{\Sigma_2}&=\Map_{\mathrm{GrMnf}}(T[1]\Sigma_2,T^*[1]M)\\
    &\cong \Omega^\bullet(\Sigma_2)\otimes T_xM\oplus \Omega^\bullet(\Sigma_2)\otimes T_x^*M[1]\ni(\boldsymbol{X},\boldsymbol{\eta}),
    \end{split}\\
    \omega_{\Sigma_2}&=\int_{\Sigma_2}\langle \delta\boldsymbol{\eta},\delta\boldsymbol{X}\rangle,\\
    \calS_{\Sigma_2}&=\int_{\Sigma_2}\langle \boldsymbol{\eta},\dd_{\Sigma_2}\boldsymbol{X}\rangle+\frac{1}{2}\int_{\Sigma_2}\langle \pi(\boldsymbol{X}),\boldsymbol{\eta}\land\boldsymbol{\eta}\rangle.
\end{align}
Choosing a formal exponential map $\phi\colon TM\to M$ together with a background field $x\colon T[1]\Sigma_2\to M$, the formal global action for the Poisson sigma model is given by 
\begin{multline}
    \Hat{\calS}^\mathrm{global}_{\Sigma_2,x}\left(\Hat{\boldsymbol{X}},\Hat{\boldsymbol{\eta}}\right)=\int_{\Sigma_2}\Hat{\boldsymbol{\eta}}_\ell\land\dd_{\Sigma_2}\Hat{\boldsymbol{X}}^\ell+\frac{1}{2}\int_{\Sigma_2}\left(\T\Tilde{\phi}_x^*\pi\right)^{ij}\left(\Hat{\boldsymbol{X}}\right)\Hat{\boldsymbol{\eta}}_i\land\Hat{\boldsymbol{\eta}}_j\\+\int_{\Sigma_2}R_{\ell}^j\left(x,\Hat{\boldsymbol{X}}\right)\Hat{\boldsymbol{\eta}}_j\land \dd_Mx^\ell.
\end{multline}

We want to construct a formal global Wilson loop like observables using the Poisson sigma model toegtehr with an auxiliary theory for an embedding $i\colon\Sigma_1:=S^1\hookrightarrow \Sigma_2$. Consider an exact symplectic manifold $(\calN,\omega_\calN=\dd_\calN\alpha_\calN)$. We can construct a vertical vector field $\calV$ on the trivial bundle $\calN\times M\to \calN$  which can be viewed as a map $\calN\to \mathfrak{X}(M)$ with the property 
\[
\frac{1}{2}\{\calV,\calV\}_{\omega_\calN}+[\pi,\calV]_{\mathrm{SN}}+R\land \calV=0,
\]
where $[\enspace,\enspace]_{\mathrm{SN}}$ denotes the \emph{Schouten--Nijenhuis} bracket defined on polyvector fields on $M$.
We have a degree 0 Hamiltonian $Q$-bundle structure on $$T^*[1]M\times \calN\to T^*[1]M$$ with fiber $\calN$ endowed with the structure  
\begin{align}
    \calV_\calE&=\langle p,\{\calV,\enspace\}_{\omega_\calN}\rangle,\\
    \Theta_\calE&=\langle p,\calV\rangle,
\end{align}
where $\calE=T^*[1]M\times\calN$. If we use the notation of Section \ref{subsec:loops_observables}, we can associate a pre-observable to the Poisson sigma model given by the data \eqref{eq:space_of_fields_loop} and \eqref{eq:omega_loop} together with the auxiliary action 
\begin{equation}
    \calS^\calN_{\Sigma_1}(\boldsymbol{X},\boldsymbol{\eta},\calB;i)=\oint_{\Sigma_1}\alpha^\calN_\nu(\calB)\dd_{\Sigma_1}\calB^\nu+\oint_{\Sigma_1}\langle i^*\boldsymbol{\eta},\calV(i^*\boldsymbol{X},\calB)\rangle.
\end{equation}
Choosing a formal exponential map $\phi\colon T\calN\to\calN$ together with local coordinates we can lift this to a formal global auxiliary action 
\begin{multline}
    \Hat{\calS}^\mathrm{global}_{\Sigma_1,y}\left(\Hat{\boldsymbol{X}},\Hat{\boldsymbol{\eta}},\Hat{\calB};i\right)=\oint_{\Sigma_1}\alpha^\calN_\nu\left(\Hat{\calB}\right)\dd_{\Sigma_1}\Hat{\calB}^\nu+\oint_{\Sigma_1}\T\Tilde{\phi}_y^*\langle i^*\boldsymbol{\eta},\calV(i^*\boldsymbol{X},\calB)\rangle\\
    +\oint_{\Sigma_1}(\Hat{R}_{\Sigma_1})_\nu\left(y,\Hat{\calB}\right)\dd_\calN y^\nu.
\end{multline}
The corresponding auxiliary formal global observable is given by 
\begin{equation}
    \Hat{\calO}_{\Sigma_1,y}\left(\Hat{\boldsymbol{X}},\Hat{\boldsymbol{\eta}};i\right)=\int_{\Hat{\calL}} \mathscr{D}\left[\Hat{\calB}_{(0)}\right]\ee^{\frac{\I}{\hbar}\Hat{\calS}^\mathrm{global}_{\Sigma_1,y}\left(\Hat{\boldsymbol{X}},\Hat{\boldsymbol{\eta}},\Hat{\calB}_{(0)};i\right)},
\end{equation}
where we use the gauge-fixing Lagrangian 
\begin{equation}
\Hat{\calL}=\Map_{\mathrm{Mnf}}(\Sigma_1,T_y\calN)\subset \Map_{\mathrm{GrMnf}}(T[1]\Sigma_1,T_y\calN)\cong\Omega^\bullet(\Sigma_1)\otimes T_y\calN.
\end{equation}

If we assume that $(\calN,\omega_\calN)$ can be geometrically quantized to a space of states $\calH$ and $\calV$ is quantized to an operator-valued vector field $\boldsymbol{\calV}\in\End(\calH)\otimes\mathfrak{X}(M)$ such that $[\pi,\boldsymbol{\calV}]_{\mathrm{SN}}+R\land\boldsymbol{\calV}+\I\hbar\boldsymbol{\calV}\land \boldsymbol{\calV}=0$, then we get that 
\begin{equation}
    \Hat{\calO}_{\Sigma_1,x}\left(\Hat{\boldsymbol{X}},\Hat{\boldsymbol{\eta}};i\right)=\tr_\calH\calP\exp\left(\frac{\I}{\hbar}\oint_{\Sigma_1}\reallywidehat{\langle i^*\boldsymbol{\eta},\boldsymbol{\calV}(i^* \boldsymbol{X})\rangle}\right),
\end{equation}
which is, by Proposition \ref{prop:formal_gloabl_wilson_loop}, indeed a formal global observable. Here we have chosen an exponential map for the base of the target of the Poisson sigma model $TM\to M$ with background field $x\in M$.

\section{Wilson surfaces and their formal global extension}

\subsection{$BF$ theory and Wilson surfaces}
Let $G$ be a Lie group and denote by $\mathfrak{g}$ its Lie algebra. Moreover, consider a principal $G$-bundle $P$ over some $d$-manifold $\Sigma_d$ and construct the adjoint bundle of $P$, denoted by $\mathrm{ad}P$, given as the frame bundle $P\times^\mathrm{Ad}\mathfrak{g}$ with respect to the adjoint representation $\mathrm{Ad}\colon G\to \mathrm{Aut}(\mathfrak{g})$ and let $\mathrm{ad}^*P$ denote its coadjoint bundle. Let $\mathscr{A}$ be the affine space of connection 1-forms on $P$ and $\mathscr{G}$ the group of gauge transformations. 
For a connection $A\in \mathscr{A}$, let $\dd_A$ be the covariant derivative on $\Omega^\bullet(\Sigma_d,\mathrm{ad}P)$ and $\Omega^\bullet(\Sigma_d,\mathrm{ad}^*P)$. 
Let $A\in \mathscr{A}$ and $B\in \Omega^{d-2}(M,\mathrm{ad}^*P)$ and define the $BF$ action by
\begin{equation}
\label{eq:BF_action}
    S(A,B):=\int_{\Sigma_d} \langle B,F_A\rangle,
\end{equation}
where $\langle\enspace,\enspace\rangle$ denotes the extension of the adjoint and coadjoint type for the canonical pairing between $\mathfrak{g}$ and $\mathfrak{g}^*$ to differential forms.
\begin{rem}[Abelian $BF$ theory]
The \emph{abelian} $BF$ action, i.e. the action for the case where $\mathfrak{g}=\R$, in fact arises as the unperturbed part of many different AKSZ theories such as the Poisson sigma model or Chern--Simons theory. In fact, for the abelian case we have $(A,B)\in\Omega^\bullet(\Sigma_d)[1]\oplus\Omega^\bullet(\Sigma_d)[d-1]$ such that $F_A=\dd A$ and thus we get an action $S=\int_{\Sigma_d}B\land \dd A$.
\end{rem}
The solutions to the Euler--Lagrange equations $\delta S=0$ for $S$ defined as in \eqref{eq:BF_action} are given by 
\begin{equation}
    \mathrm{M}_{\mathrm{cl}}=\left\{(A,B)\in \mathscr{A}\times \Omega^{d-2}(\Sigma_d,\mathrm{ad}^*P)\,\big|\, F_A=0,\dd_AB=0\right\}
\end{equation}

\begin{rem}
One can check that the $BF$ action is invariant under the action of 
\begin{equation}
    \Tilde{\mathscr{G}}:=\mathscr{G}\rtimes \Omega^{d-3}(\Sigma_d,\ad^* P),
\end{equation}
where $\mathscr{G}$ acts on $\Omega^{d-3}(\Sigma_d,\ad^* P)$ by the coadjoint action. For $(g,\sigma)\in\Tilde{\mathscr{G}}$ and $(A,B)\in \mathscr{A}\times \Omega^{d-2}(\Sigma_d,\ad^*P)$ we have an action 
\begin{equation}
    A\mapsto A^g,\qquad B\mapsto B^{(g,\sigma)}=\Ad^*_{g^{-1}}B+\dd_{A^g}\sigma.
\end{equation}
It is then easy to check that $S(A^g,B^{(g,\sigma)})=S(A,B)$.
\end{rem}
Consider an embedded submanifold $i\colon \Sigma_{d-2}\hookrightarrow \Sigma_d$ and consider the pullback bundle of $P$ by $i$ according to the diagram 
\begin{center}
\begin{tikzcd}
i^*P \arrow[r, ""] \arrow[d,]
& P \arrow[d, ""] \\
\Sigma_{d-2} \arrow[r,hook, "i"]
& \Sigma_d
\end{tikzcd}
\end{center}

We can now formulate an important type of classical action which is important for the study of higher-dimensional knots \cite{CattRoss2005}.

\begin{defn}[Wilson surface action]
The \emph{Wilson surface action} is given by 
\begin{equation}
    W(\alpha,\beta,A,B;i):=\int_{\Sigma_{d-2}}\langle \alpha,\dd_{i^*A}\beta+i^*B\rangle,
\end{equation}
where $\alpha\in \Omega^0(\Sigma_{d-2},\ad i^*P)$ and $\beta\in \Omega^{d-3}(\Sigma_{d-2},\ad^*i^*P)$. 
\end{defn}

\begin{defn}[Wilson surface observable]
The \emph{Wilson surface observable} is given by 
\begin{equation}
    \calW_{\Sigma_{d-2}}(A,B;i):=\int \mathscr{D}[\alpha]\mathscr{D}[\beta]\ee^{\frac{\I}{\hbar}W(\alpha,\beta,A,B;i)}
\end{equation}
\end{defn}

\begin{rem}
The expectation values of Wilson surface observables in fact give certain higher-dimensional knot invariants \cite{CattRoss2001}. These invariants are based on the construction of invariants by Bott \cite{Bott1996} giving the generalization to a family of isotopy invariants for long knots $\R^{n}\hookrightarrow\R^{n+2}$ for odd $n\geq 3$, which are based on constructions involving combinations of configuration space integrals. In \cite{Watanabe2007} it was proven that these invariants are of finite type for the case of long ribbon knots and that they are related to the Alexander polynomial for these type of knots. Further generalizations based on this construction, in particular for rectifiable knots, have been given in \cite{Leturcq2019,Leturcq2020}.
\end{rem}

\subsection{BV formulation of $BF$ theory}
We can consider $BF$ theory in terms of its BV extension. The BV space of fields is given by 
\begin{equation}
    \calF_{\Sigma_d}=\Omega^\bullet(\Sigma_d,\ad P)[1]\oplus\Omega^\bullet(\Sigma_d,\ad^*P)[d-2],
\end{equation}
where $\mathscr{A}=\Omega^1(\Sigma_d,\ad P)$. We will denote the superfields in $\calF_{\Sigma_d}$ by $(\boldsymbol{A},\boldsymbol{B})$. 
Note that there is an induced Lie bracket $[\![\enspace,\enspace ]\!]$ on $\Omega^\bullet(\Sigma_d,\ad P)[1]$ which is induced by the Lie bracket on $\mathfrak{g}$. 
\begin{rem}
If we consider local coordinates on $\mathfrak{g}$ with corresponding basis $(e_i)$, we have
\begin{equation}
    [\![a,b]\!]=(-1)^{\mathrm{gh}(a)\deg(b)}a^ib^jf_{ij}^ke_k,
\end{equation}
where $f_{ij}^k$ denotes the structure constants of $\mathfrak{g}$.
\end{rem}
Moreover, for $\boldsymbol{A}\in \Omega^\bullet(\Sigma_d,\ad P)[1]$ we get the curvature 
\begin{equation}
    \boldsymbol{F}_{\boldsymbol{A}}=F_{A_0}+\dd_{A_0}\boldsymbol{a}+\frac{1}{2}[\![\boldsymbol{a},\boldsymbol{a}]\!],
\end{equation}
where $A_0$ is any reference connection and $\boldsymbol{a}:=\boldsymbol{A}-A_0\in \Omega^\bullet(\Sigma_d,\ad P)[1]$. 
\begin{defn}[BV action for $BF$ theory]
\label{defn:BV_action_BF_theory}
The BV action for $BF$ theory is defined by 
\begin{equation}
    \calS_{\Sigma_d}(\boldsymbol{A},\boldsymbol{B})=\int_{\Sigma_d}\langle\!\langle \boldsymbol{B},\boldsymbol{F}_{\boldsymbol{A}}\rangle\!\rangle,
\end{equation}
where $\langle\!\langle\enspace,\enspace\rangle\!\rangle$ is the extension to forms of the adjoint and coadjoint type of the canonical pairing between $\mathfrak{g}$ and $\mathfrak{g}^*$. For two forms $a,b$ we have  
\begin{equation}
    \langle\!\langle a,b\rangle\!\rangle=(-1)^{\mathrm{gh}(a)\deg(b)}\langle a,b\rangle,
\end{equation}
\end{defn}
We can see that 
\begin{equation}
    \calF_{\Sigma_d}=T^*[-1]\Omega^\bullet(\Sigma_d,\ad P)[1],
\end{equation}
hence we have a canonical symplectic structure $\omega_{\Sigma_d}$ on $\calF_{\Sigma_d}$. Similarly as before, let us denote the odd Poisson bracket induced by $\omega_{\Sigma_d}$ by $\{\enspace,\enspace\}_{\omega_{\Sigma_d}}$ and note that $\calS_{\Sigma_d}$ satisfies the CME
\begin{equation}
    \left\{\calS_{\Sigma_d},\calS_{\Sigma_d}\right\}_{\omega_{\Sigma_d}}=0.
\end{equation}
The cohomological vector field $Q_{\Sigma_d}$ is given as the Hamiltonian vector field of $\calS_{\Sigma_d}$, thus $Q_{\Sigma_d}=\left\{\calS_{\Sigma_d},\enspace\right\}_{\omega_{\Sigma_d}}$. Note that
\begin{equation}
\label{eq:vector_field_2}
    Q_{\Sigma_d}(\boldsymbol{A})=(-1)^d\boldsymbol{F}_{\boldsymbol{A}},\qquad Q_{\Sigma_d}(\boldsymbol{B})=(-1)^d\dd_{\boldsymbol{A}}\boldsymbol{B}.
\end{equation}
If we choose a volume element $\mu$ which is compatible with $\omega_{\Sigma_d}$, we can define the BV Laplacian by 
\begin{equation}
    \Delta\colon f\mapsto \frac{1}{2}\Div_\mu\{f,\enspace\}_{\omega_{\Sigma_d}}.
\end{equation}
Then we can show that the QME holds:
\begin{equation}
    \delta_{\mathrm{BV}}\calS_{\Sigma_d}=\left\{\calS_{\Sigma_d},\calS_{\Sigma_d}\right\}_{\omega_{\Sigma_d}}-2\I\hbar\Delta\calS_{\Sigma_d}=0.
\end{equation}
This is in fact true since $\Delta\calS_{\Sigma_d}=0$. Moreover, as expected, we have $\delta_{\mathrm{BV}}^2=0$.

\subsection{Formal global $BF$ theory from the AKSZ construction}
Let us consider the case of abelian $BF$ theory. Note that in this case the Wilson surface action is given by 
\begin{equation}
    W(\alpha,\beta,A,B;i):=\int_{\Sigma_{d-2}} \alpha(\dd\beta+i^*B),
\end{equation}
where $\dd$ is the de Rham differential on $\R$. Solving the Euler--Lagrange equations for $\delta W=0$, we get that the ciritical points are solutions to 
\begin{align}
    \dd\alpha&=0,\\
    \dd \beta+i^*B&=0.
\end{align}
We want to deal with $B$ perturbatively, that means we can consider solutions to $\dd\alpha=\dd\beta=0$ instead and
hence we look at solutions of the form $\alpha=const$ and $\beta=0$. This means that the constant field $\alpha$ is going to take the place of the background field. The Wilson surface observable is then given by 
\begin{equation}
    \calW_{\Sigma_{d-2}}(A,B;i)=\int \mathscr{D}[\alpha]\mathscr{D}[\beta]\ee^{\frac{\I}{\hbar}\int_{\Sigma_{d-2}}\alpha\dd\beta}\int_{x\in\R}\mu(x)\ee^{\frac{\I}{\hbar}x\int_{\Sigma_{d-2}}i^*B},
\end{equation}
where $\mu$ is a volume element on the moduli space of classical solutions for the auxiliary theory which is given by
\begin{equation}
    \mathrm{M}_{\mathrm{cl}}=\left\{(\alpha,\beta)\in \Omega^0(\Sigma_{d-2})\oplus \Omega^{d-3}(\Sigma_{d-2})\,\big|\, \alpha=const,\, \beta=0\right\}\cong \R.
\end{equation}
By abbuse of notation we will also denote the perturbation of $\alpha$ around $x\in \R$ by $\alpha$. 
Moreover, if we assume that $P$ is a trivial bundle, not for the abelian case, we get
\begin{align}
    \calF_{\Sigma_d}&\cong \Omega^\bullet(\Sigma_d)\otimes\mathfrak{g}[1]\oplus\Omega^\bullet(\Sigma_d)\otimes \mathfrak{g}^*[d-2]\\
    &\cong \Map_{\mathrm{GrMnf}}(T[1]\Sigma_d,\mathfrak{g}[1]\oplus \mathfrak{g}^*[d-2]).
\end{align}
\begin{rem}
    The assumption that $P$ is trivial is similar to a formal lift, whereas the background field is given by a constant critical point of the form $(x,0)$ with constant background field $x\colon T[1]\Sigma_d\to \mathfrak{g}[1]\oplus \mathfrak{g}^*[d-2]$. In fact, it induces a \emph{linear split} theory as in Definition \ref{defn:split_like_AKSZ_sigma_model}.
\end{rem}

\begin{rem}[CE complex and $L_\infty$-structure]
Let $\mathfrak{g}$ be a Lie algebra and consider the differential graded algebra 
\begin{equation}
    \mathrm{CE}(\mathfrak{g}):=\left(\bigwedge^\bullet \mathfrak{g}^*,\dd_{\mathrm{CE}}\right)\cong \left(C^\infty(\mathfrak{g}[1]),Q\right).
\end{equation}
This is called the \emph{Chevalley--Eilenberg algebra} of $\mathfrak{g}$ \cite{ChevalleyEilenberg1948}. The real valued \emph{Chevalley--Eilenberg complex} is given by 
\begin{equation}
    0\rightarrow \Hom\left(\bigwedge^0\mathfrak{g},\R\right)\xrightarrow{\dd_{\mathrm{CE}}}\Hom\left(\bigwedge^1\mathfrak{g},\R\right)\xrightarrow{\dd_{\mathrm{CE}}}\Hom\left(\bigwedge^2\mathfrak{g},\R\right)\xrightarrow{\dd_{\mathrm{CE}}}\dotsm
\end{equation}
endowed with the \emph{Chevalley--Eilenberg differential} \[\dd_\mathrm{CE}\colon \Hom\left(\bigwedge^{n}\mathfrak{g},\R\right)\to \Hom\left(\bigwedge^{n+1}\mathfrak{g},\R\right)\] given by
\begin{multline}
    (\dd_\mathrm{CE}F)(X_1,\ldots,X_{n+1}):=\sum_{j=1}^{n+1}(-1)^{j+1}X_iF(X_1,\ldots,\Hat{X}_j,\ldots,X_{n+1})\\
    +\sum_{1\leq j<k\leq n+1}(-1)^{j+k}F([X_j,X_k],X_1,\ldots,\Hat{X}_j,\ldots,\Hat{X}_k,\ldots,X_{n+1}),
\end{multline}
where the hat means that these elements are omitted.
Denote by $(\xi^{i})$ the coordinates on $\mathfrak{g}[1]$ of degree $+1$. Then $Q$ has to be of the form 
\[
Q=-\frac{1}{2}f_{ij}^k\xi^{i}\xi^{j}\frac{\de}{\de \xi^{k}},
\]
where $f_{ij}^k$ are the structure constants of $\mathfrak{g}$.
Note that a function $F\in \Hom\left(\bigwedge^n\mathfrak{g},\R\right)$ corresponds to an element in $C_n^\infty(\mathfrak{g}[1])$ such that the Chevalley--Eilenberg differential is indeed mapped to $Q$ under the isomorphism 
\[
F(X_{j_1}\land\ldots\land X_{j_n})=:F_{j_1\ldots j_n}\longleftrightarrow \frac{1}{n!}\xi^{j_1}\dotsm \xi^{j_n}F_{j_1\ldots j_n}.
\]
In fact, for a graded vector space $\mathfrak{g}=\bigoplus_{k\in\mathbb{Z}}\mathfrak{g}_k$, the differential graded algebra $(C^\infty(\mathfrak{g}),Q)$ corresponds to an $L_\infty$-algebra which is actually given by the Chevalley--Eilenberg algebra $\mathrm{CE}(\mathfrak{g}[-1])$ of the $L_\infty$-algebra $\mathfrak{g}[-1]$. The dual of the cohomological vector field $Q$ is given by a codifferential $D$ of homogenous degree $+1$ on $\Hat{\Sym}(\mathfrak{g})\cong \Hat{\Sym}(\mathfrak{g}[-1])$. The isomorphism is induced by the \emph{shift isomorphism} $s\colon\mathfrak{g}\xrightarrow{\sim} \mathfrak{g}[1]$. The codifferential $D$ decomposes into a sum $D=\sum_{j\geq 1}\bar D_j$ such that the restrictions 
\[
D_j:=\bar D_j\big|_{\Hat{\Sym}^j(\mathfrak{g})}\colon \Hat{\Sym}^j(\mathfrak{g})\to \mathfrak{g}
\]
satisfy
\[
\ell_j=(-1)^{\frac{1}{2}j(j-1)+1}s^{-1}\circ D_j\circ s^{\otimes j},\qquad \forall j\geq 1.
\]
Note that since $Q^2=0$, we get $D^2=0$. Such a codifferential induces a classical Grothendieck connection as in Section \ref{subsec:notions_of_formal_geometry}.
\end{rem}

\begin{rem}[$L_\infty$-structure on $\Omega^\bullet$]
If $\mathfrak{g}$ is endowed with a (curved) $L_\infty$-structure, we can view 
\[
\Omega^\bullet(\Sigma_{d},\mathfrak{g})=\bigoplus_{\substack{r+j=k \\ 0\leq r\leq d\\ j\in\mathbb{Z}}}\Omega^r(\Sigma_d)\otimes\mathfrak{g}_{j}
\]
as a (curved) $L_\infty$-algebra. The $L_\infty$-structure arises as the linear extension of the higher brackets
\begin{align}
    \hat{\ell}_1(\alpha_1\otimes X_1)&:=\dd_{\Sigma_d}\alpha_1\otimes X_1+(-1)^{\deg(\alpha_1)}\alpha_1\otimes \ell_1(X_1)\\
    \begin{split}
    \hat{\ell}_n(\alpha_1\otimes X_1,\ldots,\alpha_n\otimes X_n)&:=(-1)^{n\sum_{j=1}^{n}\deg(\alpha_j)+\sum_{j=0}^{n-2}\deg(\alpha_{n-j})\sum_{k=1}^{n-j-1}\deg(X_k)}\times\\
    &\times (\alpha_1\land\dotsm\land \alpha_n)\otimes \ell_n(X_1,\ldots,X_n)
    \end{split}
\end{align}
for $n\geq 2$, $\alpha_1,\ldots,\alpha_n\in\Omega^\bullet(\Sigma_d)$ and $X_1,\ldots,X_n\in\mathfrak{g}$. If $\mathfrak{g}$ is cyclic, and $\Sigma_d$ is compact, oriented without boundary, there is a natural cyclic inner product on $\Omega^\bullet(\Sigma_d,\mathfrak{g})$ given by
\begin{equation}
    \langle \alpha_1\otimes X_1,\alpha_2\otimes X_2\rangle_{\Omega^\bullet(\Sigma_d,\mathfrak{g})}=(-1)^{\deg(\alpha_2)\deg(X_1)}\int_{\Sigma_d}\alpha_1\land\alpha_2\langle X_1,X_2\rangle_{\mathfrak{g}}
\end{equation}
for $\alpha_1,\alpha_2\in\Omega^\bullet(\Sigma_d)$ and $X_1,X_2\in\mathfrak{g}$.
\end{rem}

\subsection{BV extension of Wilson surfaces}
We will now construct the BV extended observable for the auxiliary codimension 2 theory in the case where $P$ is a trivial bundle. Let 
\begin{equation}
    \calF_{\Sigma_{d-2}}\cong \Omega^\bullet(\Sigma_{d-2})\otimes\mathfrak{g}[1]\oplus\Omega^\bullet(\Sigma_{d-2})\otimes \mathfrak{g}^*[d-2]
\end{equation}
endowed with the symplectic form $\omega_{\Sigma_{d-2}}$ which induces the corresponding BV bracket $\{\enspace,\enspace\}_{\omega_{\Sigma_{d-2}}}$. For auxiliary superfields $(\boldsymbol{\alpha},\boldsymbol{\beta})\in \calF_{\Sigma_{d-2}}$ and ambient superfields $(\boldsymbol{A},\boldsymbol{B})\in\calF_{\Sigma_d}$ we have the following definition:
\begin{defn}[BV extended Wilson surface action]
The \emph{BV extended Wilson surface action} is given by 
\begin{equation}
\label{eq:BV_extended_Wilson_surface_action}
    \boldsymbol{W}^0_{\Sigma_{d-2}}(\boldsymbol{\alpha},\boldsymbol{\beta},\boldsymbol{A},\boldsymbol{B};i)=\int_{\Sigma_{d-2}}\langle\!\langle \boldsymbol{\alpha},\dd_{i^*\boldsymbol{A}}\boldsymbol{\beta}+i^*\boldsymbol{B}\rangle\!\rangle.
\end{equation}
\end{defn}
\begin{rem}
As it was shown in \cite{CattRoss2005}, we can extend $\boldsymbol{W}^0_{\Sigma_{d-2}}$, regarded as a function on embeddings $i\colon \Sigma_{d-2}\hookrightarrow \Sigma_d$, to a form-valued function $\boldsymbol{W}_{\Sigma_{d-2}}$ on these embeddings by setting
\begin{equation}
    \boldsymbol{W}_{\Sigma_{d-2}}(\boldsymbol{\alpha},\boldsymbol{\beta},\boldsymbol{A},\boldsymbol{B};i):=\pi_*\langle\!\langle \boldsymbol{\alpha},\dd_{\mathrm{ev}^*\boldsymbol{A}}\boldsymbol{\beta}+\mathrm{ev}^*\boldsymbol{B}\rangle\!\rangle,
\end{equation}
where $\mathrm{ev}$ denotes the evaluation map of embeddings $\Sigma_{d-2}\hookrightarrow\Sigma_d$ and $\pi_*$ denotes the integration along the fiber $\Sigma_{d-2}$.
\end{rem}
\begin{prop}[\cite{CattRoss2005}]
\label{prop:CattRoss1}
The Wilson surface action satisfies a modified version of the dCME, i.e. we have 
\begin{equation}
    \dd\boldsymbol{W}_{\Sigma_{d-2}}-(-1)^d\left\{\boldsymbol{W}_{\Sigma_{d-2}},\boldsymbol{W}_{\Sigma_{d-2}}\right\}_{\omega_{\Sigma_d}}-\frac{1}{2}\left\{\boldsymbol{W}_{\Sigma_{d-2}},\boldsymbol{W}_{\Sigma_{d-2}}\right\}_{\omega_{\Sigma_{d-2}}}=0
\end{equation}
\end{prop} 
\begin{rem}
Proposition \ref{prop:CattRoss1} is a consequense of the fact that $$\dd\int_{\Sigma_{d-2}}=(-1)^d\int_{\Sigma_{d-2}}\dd$$ and \eqref{eq:vector_field_2}.
\end{rem}
Denote by $\Delta_{\Sigma_{d-2}}$ the BV Laplacian for the auxiliary theory. Then we get the following proposition.
\begin{prop}[\cite{CattRoss2005}]
Define the vector field 
\begin{equation}
\boldsymbol{Q}_{\Sigma_{d-2}}=\left\{\boldsymbol{W}_{\Sigma_{d-2}},\enspace\right\}_{\omega_{\Sigma_{d-2}}},
\end{equation}
which acts on generators by 
\begin{equation}
\label{eq:vector_fields}
    \boldsymbol{Q}_{\Sigma_{d-2}}(\boldsymbol{\alpha})=(-1)^d\dd_{\mathrm{ev}^*\boldsymbol{A}}\boldsymbol{\alpha},\qquad \boldsymbol{Q}_{\Sigma_{d-2}}(\boldsymbol{\beta})=(-1)^d(\dd_{\mathrm{ev}^*\boldsymbol{A}}\boldsymbol{\beta}+\mathrm{ev}^*\boldsymbol{B}).
\end{equation}
Assume that the formal measure $\mathscr{D}[\boldsymbol{\alpha}]\mathscr{D}[\boldsymbol{\beta}]$ is invariant with respect to the vector fields \eqref{eq:vector_fields}.
Then we have
\begin{multline}
\dd\boldsymbol{W}_{\Sigma_{d-2}}-(-1)^d\left(\delta_{\mathrm{BV}}\boldsymbol{W}_{\Sigma_{d-2}}+\frac{1}{2}\left\{\boldsymbol{W}_{\Sigma_{d-2}},\boldsymbol{W}_{\Sigma_{d-2}}\right\}_{\omega_{\Sigma_d}}\right)\\
+\frac{1}{2}\left(\left\{\boldsymbol{W}_{\Sigma_{d-2}},\boldsymbol{W}_{\Sigma_{d-2}}\right\}_{\omega_{\Sigma_{d-2}}}-2\I\hbar\Delta_{\Sigma_{d-2}}\boldsymbol{W}_{\Sigma_{d-2}}\right)=0
\end{multline}
\end{prop}
\begin{rem}
Note that the assumption of invariance of the formal measure implies that $\Delta_{\Sigma_{d-2}}\boldsymbol{W}_{\Sigma_{d-2}}=0$. 
\end{rem}
\begin{defn}[BV extended Wilson surface observable]
We define the \emph{BV extended Wilson surface observable} as 
\begin{equation}
    \boldsymbol{\calW}_{\Sigma_{d-2}}(\boldsymbol{A},\boldsymbol{B};i)=\int \mathscr{D}[\boldsymbol{\alpha}]\mathscr{D}[\boldsymbol{\beta}]\ee^{\frac{\I}{\hbar}\boldsymbol{W}_{\Sigma_{d-2}}(\boldsymbol{\alpha},\boldsymbol{\beta},\boldsymbol{A},\boldsymbol{B};i)}
\end{equation}
\end{defn}

\subsection{Formulation by Hamiltonian $Q$-bundles}
Let $\calM=\mathfrak{g}[1]\oplus \mathfrak{g}^*[d-2]$. Denote by $x\colon \mathfrak{g}[1]\to \mathfrak{g}$ the degree 1 $\mathfrak{g}$-valued coordinate on $\mathfrak{g}[1]$ and let $x^*\colon \mathfrak{g}^*[d-2]\to \mathfrak{g}^*$ be the $\mathfrak{g}^*$-valued coordinate on $\mathfrak{g}^*[d-2]$ of degree $d-2$.

Then we can consider a trivial Hamiltonian $Q$-bundle over $\calM$ given by the fiber data
\begin{align}
    \calN&=\mathfrak{g}\oplus\mathfrak{g}^*[d-3],\\
    \calV_\calE&=\left\langle [x,y],\frac{\de}{\de y}\right\rangle+\left\langle \ad_x^*y^*,\frac{\de}{\de y^*}\right\rangle+(-1)^d\left\langle x^*,\frac{\de}{\de y^*}\right\rangle,\\
    \omega_\calN&=\left\langle \delta y^*,\delta y\right\rangle,\\
    \alpha_\calN&=\langle y^*,\delta y\rangle,\\
    \Theta_{\calE}&=\left\langle y^*,[x,y]\right\rangle+\langle x^*,y\rangle,
\end{align}
where $y$ is the $\mathfrak{g}$-valued coordinate of degree 0 on $\mathfrak{g}$ and $y^*$ is the $\mathfrak{g}^*$-valued coordinate of degree $d-3$ on $\mathfrak{g}^*[d-3]$. For an embedding $i\colon \Sigma_{d-2}\hookrightarrow \Sigma_d$ we get the auxiliary theory
\begin{align}
    \calF^\calN_{\Sigma_{d-2}}&=\Omega^\bullet(\Sigma_{d-2})\otimes \mathfrak{g}\oplus \Omega^\bullet(\Sigma_{d-2})\otimes \mathfrak{g}^*[d-3],\\
    \omega^\calN_{\Sigma_{d-2}}&=(-1)^d\int_{\Sigma_{d-2}}\langle \delta\boldsymbol{y}^*,\delta\boldsymbol{y}\rangle,\\
    \calS^\calN_{\Sigma_{d-2}}&=\int_{\Sigma_{d-2}}\langle \boldsymbol{y}^*,\dd_{\Sigma_{d-2}}\boldsymbol{y}\rangle+\int_{\Sigma_{d-2}}\langle \boldsymbol{y}^*,[i^*\boldsymbol{A},\boldsymbol{y}]\rangle+\int_{\Sigma_{d-2}}\langle i^*\boldsymbol{B},\boldsymbol{y}\rangle.
\end{align}
Note that $\calM$ is a differential graded symplectic manifold with the following data:
\begin{align}
    Q_\calM&=\left\langle \frac{1}{2}[x,x],\frac{\de}{\de x}\right\rangle+\left\langle \ad_x^*x^*,\frac{\de}{\de x^*}\right\rangle,\\
    \omega_\calM&=\langle \delta x^*,\delta x\rangle,\\
    \alpha_\calM&=\langle x^*,\delta x\rangle,\\
    \Theta_\calM&=\frac{1}{2}\left\langle x^*,[x,x]\right\rangle.
\end{align}
Hence the ambient theory is given by 
\begin{align}
    \calF^\calM_{\Sigma_d}&=\Omega^\bullet(\Sigma_d)\otimes \mathfrak{g}[1]\oplus \Omega^\bullet(\Sigma_d)\otimes \mathfrak{g}^*[d-2]\ni (\boldsymbol{A},\boldsymbol{B}),\\
    \omega_{\Sigma_d}&=\int_{\Sigma_d}\langle\delta \boldsymbol{B},\delta\boldsymbol{A}\rangle,\\
    \label{eq:BF_action}
    \calS_{\Sigma_d}&=\int_{\Sigma_d}\left\langle \boldsymbol{B},\boldsymbol{F}_{\boldsymbol{A}}\right\rangle=\int_{\Sigma_d}\left\langle \boldsymbol{B},\dd_{\Sigma_d}\boldsymbol{A}+\frac{1}{2}[\boldsymbol{A},\boldsymbol{A}]\right\rangle.
\end{align}
Note that \eqref{eq:BF_action} is exactly the $BF$ action as in Definition \ref{defn:BV_action_BF_theory}.
In the case of abelian $BF$ theory, i.e. when $\mathfrak{g}=\R$, we get that $Q_\calM=0$, $\Theta_\calM=0$ and the ambient theory
\begin{align}
    \calF^\calM_{\Sigma_d}&=\Omega^\bullet(\Sigma_d)[1]\oplus \Omega^\bullet(\Sigma_d)[d-2],\\
    \omega_{\Sigma_d}&=\int_{\Sigma_d}\delta\boldsymbol{B}\land \delta\boldsymbol{A},\\
    \calS_{\Sigma_d}&=\int_{\Sigma_d}\boldsymbol{B}\land \dd_{\Sigma_d}\boldsymbol{A}.
\end{align}

\begin{rem}
The constructions presented in this paper are expected to extend to manifolds with boundary. Using the constructions as in \cite{CMW3} together with the quantum BV-BFV formalism \cite{CMR2,CattMosh1}, one can show how the formal global observables for split AKSZ sigma models on the boundary induce a more general gauge condition as the dQME which is called \emph{modified differential Quantum Master Equation (mdQME)}. This condition also handles the boundary part which arises as the ordered standard quantization $\Omega$ of the boundary action $\calS^\de$ of dgeree $+1$, induced by the underlying BFV manifold, plus some higher degree terms. The mdQME is then given as some annihilation condition for the formal global boundary observable $\calO^\de$. In fact, it is annihilated by the quantum Grothendieck BFV operator $\nabla_\mathsf{G}:=\dd-\I\hbar\Delta+\frac{\I}{\hbar}\Omega$ (see \cite{CMW3,CMW4}), which means that $\nabla_\mathsf{G}\calO^\de=0$.
\end{rem}

\printbibliography

\end{document}